\documentclass[a4paper,11pt]{article}

\usepackage{a4wide}
\usepackage{hyperref}
\usepackage{authblk}
\usepackage{dsfont}
\usepackage{color}

\usepackage[round]{natbib}

\makeatletter
\providecommand{\doi}[1]{
  \begingroup
    \let\bibinfo\@secondoftwo
    \urlstyle{rm}
    \href{http://dx.doi.org/#1}{
      doi:\discretionary{}{}{}
      \nolinkurl{#1}
    }
  \endgroup
}
\makeatother

\usepackage{amsthm, amssymb, amsmath}

\usepackage{subfig}

\usepackage{verbatim}
\usepackage{graphicx}

\newcommand{\R}{\ensuremath{\mathds{R}}}

\newcommand{\Z}{\ensuremath{\mathds{Z}}}

\newcommand{\Obj}{f}

\newcommand{\ObjFeas}{Y}
\newcommand{\ObjSub}{Q}

\newcommand{\ND}{{\ObjFeas}_{\rm nd}}

\newcommand{\bz}{\bar{z}}
\newcommand{\bu}{\bar{u}}
\newcommand{\hN}{\hat{N}}
\newcommand{\hZ}{\hat{Z}}
\newcommand{\BM}{\boldsymbol{M}}
\newcommand{\Bm}{\boldsymbol{m}}
\newcommand{\Bzero}{\boldsymbol{0}}
\newcommand{\Bone}{\boldsymbol{1}}

\newcommand{\onetod}{\{1,\dots,p\}}
\newcommand{\dummy}{\hat{z}}

\newcommand{\GP}{$\textrm{SA}$}

\newcommand{\GPl}{Simplifying assumption}
\newcommand{\NGPl}{No simplifying assumption}

\newcommand{\DocType}{paper}

\newcommand{\lub}{local upper bound}
\newcommand{\ubs}{upper bound set}

\usepackage[ruled,vlined,linesnumbered]{algorithm2e}
\SetCommentSty{emph}
\SetKwComment{tcp}{-- }{}
\SetKwComment{tcc}{-- }{ --}
\SetEndCharOfAlgoLine{}

\usepackage{diagbox}

\newtheorem{lemma}{Lemma}[section]

\newtheorem{theorem}[lemma]{Theorem}
\newtheorem{remark}[lemma]{Remark}

\newtheorem{definition}[lemma]{Definition}

\newtheorem{proposition}[lemma]{Proposition}

\theoremstyle{definition}
\newtheorem{example}{Example}

\usepackage{enumitem}
\setlist{leftmargin=2\parindent}

\begin{document}

\title{On the representation of the search region\\
in multi-objective optimization\footnote{This work 
was supported by French ANR-09-BLAN-0361 
"GUaranteed Efficiency for PAReto optimal solutions Determination (GUEPARD)"}
\footnote{This work is to appear in \emph{European Journal of Operational Research} 
and is available online at \url{http://www.sciencedirect.com/science/article/pii/S0377221715002386}}
}

\author[1]{Kathrin Klamroth}
\author[2]{Renaud Lacour}
\author[2]{Daniel Vanderpooten}
\affil[1]{Department of Mathematics and Computer Science, 
University of Wuppertal, Germany\\
\texttt{kathrin.klamroth@math.uni-wuppertal.de}}
\affil[2]{PSL, Universit\'e Paris-Dauphine, LAMSADE UMR 7243, F-75016 Paris, France\\
\texttt{\{lacour,vdp\}@lamsade.dauphine.fr}}
\maketitle

\begin{abstract}
Given a finite set $N$ of feasible points
of a multi-objective optimization (MOO) problem,
the \emph{search region} corresponds to the part of the objective space
containing all the points that are not dominated by any point of $N$,
i.e. the part of the objective space which may contain further nondominated points.
In this {\DocType}, we consider a representation of the search region
by a set of tight \emph{\lub{}s} (in the minimization case)
that can be derived from the points of $N$.
Local upper bounds play an important role in methods
for generating or approximating the nondominated set of an MOO problem,
yet few works in the field of MOO address their efficient incremental determination.
We relate this issue to the state of the art in computational geometry
and provide several equivalent definitions of \lub{}s that are meaningful in MOO.
We discuss the complexity of this representation in arbitrary dimension, 
which yields an improved upper bound on the number of solver calls 
in epsilon-constraint-like methods to generate the nondominated set 
of a discrete MOO problem.
We analyze and enhance a first incremental approach 
which operates by eliminating redundancies among \lub{}s.
We also study some properties of \lub{}s, especially concerning the issue
of redundant \lub{}s,
that give rise to a new incremental approach which avoids such redundancies.
Finally, the complexities of the incremental approaches are compared 
from the theoretical and empirical points of view.

\end{abstract}

\section{Introduction}
Most solution approaches in multi-objective optimization (MOO)
aimed at outputting a set of ``good'' solutions
iteratively generate candidate solutions.
Generally, a pool of solutions is maintained and updated 
when new solutions arrive.
The pool provides information which is used to decide 
whether a new solution should be inserted and whether old solutions should be removed.
It can also be used to guide the search process within the objective space.
In particular, from the images in the objective space of the pool solutions,
we can define the part of the objective space containing all points 
that none of these images dominate, 
which we refer to as the \emph{search region}.

The concept itself is well known in the field. 
Especially in the two dimensional or bi-objective case,
it is a key tool of the two-phase and branch and bound methods. 
In the two-phase method \citep[see][]{UluTeg95}, adjacent extreme nondominated points
computed in the first phase define triangles 
which delimit zones where all other nondominated points lie.
The so-called \emph{local nadir} points 
corresponding to the right angles of these triangles act as 
\emph{local upper bounds} that together define the search region,
assuming that the objectives are to be minimized. 
This upper bounding part is also one of the foundations 
of multi-objective branch and bound \citep[see][]{SouSpa08a}. 
Actually, the representation of the search region through local upper bounds
makes it possible to test the existence of the intersection 
between the search region 
and a convex lower bound 
on the feasible points associated to a search node 
and to decide whether to fathom the search node or not.

Given a discrete set of points $Q\subset \R^p$, \citet{KapRubShaVer08}
consider \emph{maximal empty orthants} with respect to $Q$, 
which contain no point of $Q$ and are maximal for this property under inclusion.
Assuming that the points of $Q$ are feasible points of an MOO problem, 
the union of all maximal empty orthants corresponds
to the search region 
associated to $Q$
and their apexes correspond to local upper bounds in the context of MOO.
\citet{KapRubShaVer08} give an algorithm 
for the generation of all maximal empty orthants, and hence
for the computation of all local upper bounds of the search region.
However it requires that the input points are given 
in a nondecreasing order of some component
and in this sense does not directly imply an incremental approach.

\citet{PrzGanEhr10} consider local upper bounds in arbitrary dimension
for a generalization 
of the two-phase method to problems with arbitrarily many objectives.
They propose an online algorithm to carry out the update 
of the local upper bounds as soon as new feasible points are discovered,
but they do not consider the complexity of this operation.

Several solution methods to generate the nondominated set 
iteratively solve linear programs parameterized by local upper bounds
\citep[see][]{SylCre07}, possibly including redundancies \citep[see][]{LokKok12, KirSay14}.
In \citet{SylCre07} 
each \lub{} is determined by solving an integer linear program.

Also, \citet{DaeKla14} propose to compute \emph{boxes}
for three dimensional MOO problems 
that are defined by a common lower bound and several upper bounds 
and decompose the search region.
They develop an efficient incremental algorithm,
that avoids redundancies,
to update the decomposition each time a nondominated point is found,
through e.g. the optimization of a pseudo-distance function 
parameterized by the defining points of the box.
In particular, they show that,
in the three dimensional case,
the search region can be described
by $2n+1$ boxes if the number of known feasible points is $n$.

The {\DocType} is organized as follows. 
Section~\ref{sec:formulation_def} sets some notations, formally defines 
the concepts of \emph{search region} and \emph{upper bound set}, 
then motivates their use in MOO. 
Section~\ref{sec:existence} shows the existence and uniqueness of \lub{}s
through a first algorithm for which we discuss some enhancements.
Section~\ref{sec:properties} investigates some properties 
of the elements of upper bound sets that yield another approach
to compute an upper bound set.
Section~\ref{sec:complexity_experiments} is devoted to the complexity aspects
related to the representation of the search region by a set of \lub{}s 
and the comparison of the two approaches 
from both theoretical and empirical points of view.
Finally, Section~\ref{sec:conclusions} provides conclusions and perspectives.

\section{Background and motivations}\label{sec:formulation_def}

\subsection{Multi-objective optimization setting}

We consider MOO problems 
\begin{equation}\label{eq:MOP}
\begin{array}{ll} 
\min        & \Obj(x)=(\Obj_1(x),\dots,\Obj_p(x)) \\
\mbox{s.t.} & x \in X
\end{array}
\end{equation}
with feasible set 
$X \neq \emptyset$
and with $p \geq 2$ objective functions $\Obj_j:X \to \R$, $j\in\onetod$.
Let $\ObjFeas=\Obj(X)$ denote the set of all feasible points in the objective space.
For all $j\in\onetod$ and $x\in X$, 
we assume, for any instance of an MOO problem,
that $m<\Obj_j(x)<M$ for some $m,M\in \R$,
or that such values $m$ and $M$ exist that bound the area of interest 
for the decision maker.
We will refer to $Z={(m,M)^p}$ as the $p$-dimensional
\emph{search interval},
a set that contains all feasible or at least all relevant points.
We denote by $\hZ=[m,M]^p$ the closure of $Z$.

The Pareto concept of optimality for MOO problems 
is based on the componentwise orderings of $\R^{p}$ 
defined for $ z^1, z^2 \in \R^{p}$ by
$$
\begin{array}{cccc}
z^1 \leqq z^2 & \text{($z^1$ weakly dominates $z^2$)} & \Leftrightarrow & z^1_j \leq z^2_j, \quad \
  j\in\onetod ,\\
z^1 \leq z^2  & \text{($z^1$ dominates $z^2$)} & \Leftrightarrow & z^1 \leqq z^2 
\quad \mbox{and} \quad z^1 \neq z^2, \\
z^1 < z^2     & \text{($z^1$ strictly dominates $z^2$)} & \Leftrightarrow & z^1_{j} < z^2_{j}, \quad  \ j
= 1,\dots, p.
\end{array}
$$
A point $z^2 \in \R^{p}$ is called {\em dominated} by $z^1 \in \R^{p}$ 
if $ z^1 \leq z^2$. 
If, moreover, $ z^1 < z^2$ then $z^2$ is called 
{\em strictly dominated} by $z^1$.
A subset $N$ of $Z$ is \emph{stable for the dominance relation $\leq$} 
or simply \emph{stable} 
if for any $z^1,z^2\in N$, $z^1\not\leq z^2$.
For any subset $\ObjSub$ of $\R^{p}$, $\ObjSub_{\rm nd}$ is the set 
of all nondominated points of $\ObjSub$, i.e. 
$\ObjSub_{\rm nd}= \{z \in \ObjSub: \mbox{there exists no } \bz \in \ObjSub 
\mbox{ with } \bz \leq z \}$.
We refer to $\ND$ as {\em the nondominated set} of \eqref{eq:MOP},
and every point $z\in\ND$ is called \emph{nondominated}.

Note that $M$ must therefore be strictly greater than the component values
of the nadir point, defined as the componentwise maximum
of the nondominated points, i.e. $(\max\{z_j : z\in \ND\})_{j\in\onetod}$.

We define some general notations. 
We denote by $\BM$ the $p$-dimensional vector $(M,\dots,M)$ 
and analogously the $p$-dimensional vector $\Bm=(m,\dots,m)$ 
and the $p$-dimensional all-ones vector $\Bone$.
For any $z\in\R^p$, we let $z_{-j}$ be the $(p-1)$-dimensional
vector of all components of $z$ excluding component~$j$, for a given
$j\in\{1,\dots,p\}$. 
Finally, for any $z,a\in\R^{p}$ and any $j\in\{1,\dots,p\}$, $(z_{j},a_{-j})$
denotes the vector $(a_{1},\dots,a_{j-1},z_{j},a_{j+1},\dots,a_{p})$. 
Such a vector will be referred to 
as the $j$th projection of vector~$z$ on vector~$a$.

\subsection{The search region}\label{sub:the_search_region}

In the following definition, we formalize  the concept of \emph{search region}
which we presented in the introduction.

\begin{definition}\label{def:sr}
Let $N$ be a finite and stable set of feasible points.
The \emph{search region} for $\ND \setminus N$, denoted by $S(N)$,
contains all the points 
in $Z$ that could be nondominated given $N$, 
or alternatively, excludes all the points in $Z$ 
that are dominated by at least one point in $N$, that is:
 \begin{equation}\label{eq:searchregion}
\begin{array}{rcl}
 S(N) & = & \{ z \in  Z\,:\, \forall\bz\in N,\, 
  \bz\not\leqq z\}\\
 & = &  Z \setminus \{ z\in Z \,:\, \exists \bz\in N \text{~with~} \bz\leqq z\}
\end{array}
\end{equation}
\end{definition}

Note that the search region $S(N)$ excludes the points in $N$,
since they are already known, that is, $S(N)\cap N = \emptyset$.

In some cases, $N$ is a subset of $\ND$ obtained by some 
scalarizing function. 
More generally, $N$ may contain any feasible point, 
no matter how it is obtained, e.g. by any heuristic procedure.
It could even be any stable set of not necessarily feasible points 
from $Z$, provided none of its points dominate
any point of the nondominated set $\ND$.

Regarding the stability condition on $N$,
it can be easily seen that the set $S(N)$ is not affected 
if a point dominated by another point of $N$ is added. In other words:
\begin{remark}
For any set of points $\ObjSub$, we have $S(\ObjSub)=S(\ObjSub_{\rm nd})$, 
i.e., both sets induce the same search region.
\end{remark}
\noindent Consequently, the assumption that the set $N$ is stable 
can be made without loss of generality.

\subsection{Explicit representation of the search region by \lub{}s}\label{sub:explicit}

Our purpose is to find an explicit and concise 
characterization of $S(N)$ using a finite set $U(N)$ of minimal \emph{\lub{}s},
which could also be referred to as \emph{local nadir points} 
or \emph{maximal points} (for the dominance relation).
We will refer to $U(N)$ as an \emph{\ubs{}} for the search region $S(N)$ in the following.

Every \lub{} $u\in U(N)$ defines a \emph{search zone} $C(u)\subset Z$ as
$$C(u)=\{z\in Z\,:\, z< u\},$$
and the search region $S(N)$ is covered by the union of these search zones.
In order to possibly include any point of $Z$ in a given search zone $C(u)$,
the possible values for $u$ 
should include the boundary of $Z$. 
So in general $U(N)$ is a subset of $\hZ$, the closure of the search interval $Z$.

In the following, we give three alternative definitions for \ubs{}s 
and show their equivalence.

\begin{definition}\label{def:upperbounds1}
Let $N\subset Z$ be a finite and stable set of points. 
A set $U(N)\subset \hZ$ is called an \emph{\ubs{} 
with respect to $N$} if and only if
\begin{itemize}
 \item[(1)] $S(N)= \bigcup_{u\in U(N)} C(u)$ and
 \item[(2)] $\forall u^1,u^2\in U(N),\,C(u^1)\not\subset C(u^2)$.
\end{itemize}
\end{definition}

While condition (1) in Definition \ref{def:upperbounds1} guarantees that the search region $S(N)$ is
\emph{exactly} represented by the search zones induced by $U(N)$,
condition (2) ensures minimality of the set $U(N)$ in the sense that no redundant search zones are contained in the representation.
Observe that this definition can be seen as a natural extension of the concept of \emph{upper bound} in the one-dimensional case.
If $p=1$, a stable set may either be empty ($N=\emptyset$), or it may consist of exactly one point ($N=\{\bz\}$). The corresponding search region is then uniquely represented by one point, namely
$\bar{u}=\BM$ in the first case and $\bar{u}=\bar{z}$ in the latter case.

As an example, we describe the situation in the two-dimensional case.
\begin{example}\label{ex:bi_obj}
Let $N=\{(z_{1}^{1},z_{2}^{1}),\dots,$ $(z_{1}^{n},z_{2}^{n})\}$ be
a stable set of two-dimensional points (with $n \geq 1$). 
In the bi-objective case, 
the points in any stable set $N$ can be ordered 
such that the objective values are strictly increasing in the first objective 
and strictly decreasing in the second objective.
Hence we can assume that $z_{1}^{1}<\dots<z_{1}^{n}$
and $z_{2}^{1}>\dots>z_{2}^{n}$. 
The search region consists of the union of search zones 
defined by pairs of consecutive points in $N$.
Thus the upper bound set associated to $N$ is
$$U(N)=\left\{\left(z_{1}^{1},M\right), \left(z_{1}^{2},z_{2}^{1}\right), \left(z_{1}^{3},z_{2}^{2}\right),
\dots, \left(z_{1}^{n},z_{2}^{n-1}\right), \left(M,z_{2}^{n}\right)\right\}$$

We illustrate this example 
in Figure~\ref{fig:illus_def_1}.
\begin{figure}[!htbp]
\begin{centering}
\includegraphics{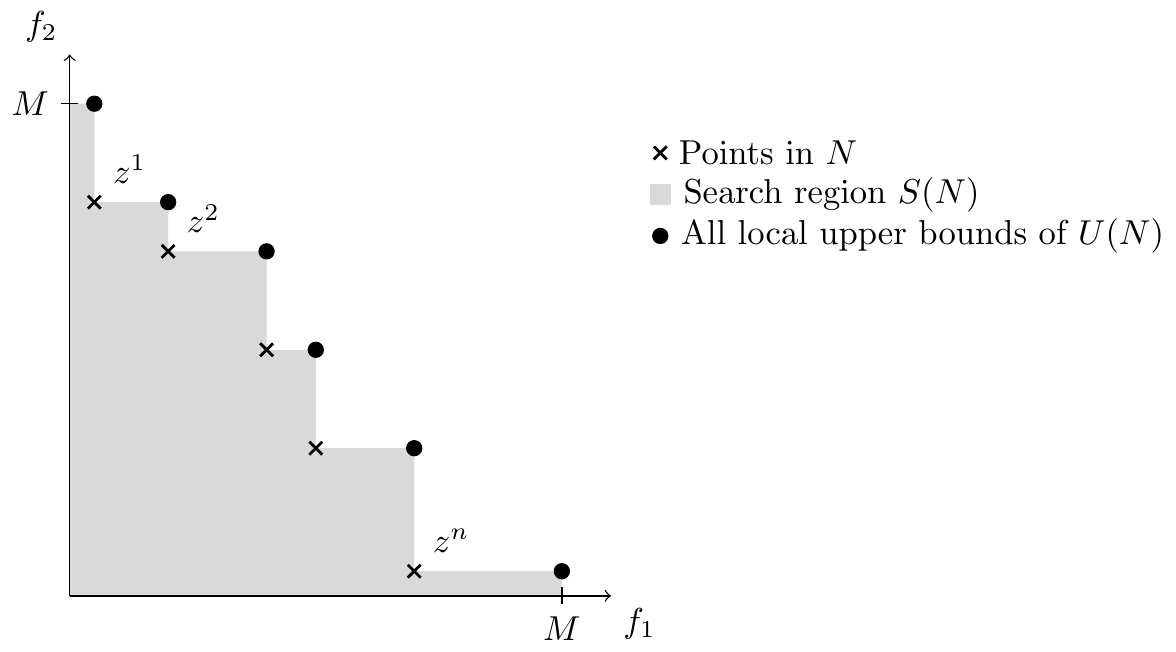}
\par\end{centering}

\caption{Illustration of the concepts of 
\emph{search region} and \emph{\lub{}} ($p=2$)\label{fig:illus_def_1}}
\end{figure}
\end{example}

The conditions of Definition \ref{def:upperbounds1} 
can immediately be reformulated 
in terms of pairwise comparisons 
between points in $S(N)$ and $U(N)$.

\begin{proposition}\label{def:upperbounds2}
Let $N\subset Z$ be a finite and stable set of points. 
Then $U(N)\subset\hZ$ is an \ubs{} with respect to $N$ 
if and only if
\begin{itemize}
 \item[(1a')] $\forall z\in S(N) \,\exists u\in U(N) \, :\, z<u$, 
 \item[(1b')] $\forall z\in Z\setminus S(N)\, \forall u\in U(N) \,:\, z\not< u$, and
 \item[(2')] $\forall u^1,u^2\in U(N) \,:\, u^1\not\leq u^2$.
\end{itemize}
\end{proposition}

\begin{proof}
Conditions (1a') and (1b') together are equivalent 
to condition (1) in Definition \ref{def:upperbounds1}, 
and condition (2') is equivalent 
to condition (2) in Definition \ref{def:upperbounds1}.
\end{proof}

In the case where the search interval is restricted to integer-valued vectors,
i.e. $Z\subset\Z^p$, conditions (1a') and (1b') of Proposition~\ref{def:upperbounds2}
can be further specified since for all $z,z'\in\Z^p$ such that $z<z'$, 
we have $z\leqq z'-\Bone$. 
We briefly restate them in the following remark.

\begin{remark}\label{rem:upperbounds2}
 Assume $Z\subset\Z^p$ and $\BM\in \Z^p$. Under the same hypothesis 
 of Proposition~\ref{def:upperbounds2}, we have
 for the \ubs{} $U(N)\subset \hZ$:
\begin{itemize}
 \item[(1a'')] $\forall z\in S(N) \,\exists u\in U(N) \, :\, z\leqq u-\Bone$ and
 \item[(1b'')] $\forall z\in Z\setminus S(N)\, \forall u\in U(N) \,:
 \, z\not\leqq u-\Bone$
\end{itemize}
\end{remark}

This is particularly useful in the context of the two-phase 
and branch and bound algorithms which we discuss at the end of this section.

A yet alternative characterization of \lub{}s 
that will turn out useful for their efficient determination 
is given in Proposition~\ref{def:upperbounds3}. 
In particular, \lub{}s are exactly those points that 
(i) are not strictly dominated by any of the points in $N$, and 
(ii) are maximal with this property.

\begin{proposition}\label{def:upperbounds3}
 Let $N\subset Z$ be a finite and stable set of points. Then $U(N)\subset \hZ$ 
 is an \ubs{} with respect to $N$ if and only if $U(N)$ consists of all points $u\in \hZ$ that satisfy the following two conditions:
 \begin{itemize}
  \item[(i)] no point of $N$ strictly dominates $u$ and
  \item[(ii)] for any $\bu\in\hZ$ such that $\bu \geq u$, there exists $\bz\in N$ such that $\bz<\bu$, i.e., $u$ is a maximal point with property (i).
 \end{itemize}
\end{proposition}

\begin{proof} 
Let $U(N)$ denote {the \ubs{}} with respect to $N$, and let 
$U'(N)$ denote the set of all points satisfying (i) and (ii) above.\\ 

  \emph{Claim 1:} $U(N)\subset U'(N)$.
  Let $u\in U(N)$. We show that $u$ satisfies (i) and (ii).
  \begin{itemize}
  \item[(i)] Assume that there exists a point $\bar{z}\in N$ 
  such that $\bar{z}< u$. 
  Then by condition (1)  of Definition~\ref{def:upperbounds1}, 
  $\bar{z}\in C(u)\subset S(N)$, 
  which contradicts $S(N)\cap N=\emptyset$.
  \item[(ii)] Let $\bar{u}\geq u$ and hence $C(u)\subset C(\bar{u})$. 
  Since $u\in U(N)$, by condition (2) of Definition~\ref{def:upperbounds1},
  we get $\bu\not\in U(N)$.
  Thus, there exists a point $z'\in C(\bar{u})$ 
  such that $z'\in Z\setminus S(N)$. 
  Since $z'\in Z\setminus S(N)$, 
  there exists a point $\bar{z}\in N$ with $\bar{z}\leqq z'$ 
  and hence $\bar{z}<\bar{u}$. 
  \end{itemize}
  
  \emph{Claim 2:} $U'(N)= U(N)$, 
i.e., we show that $U'(N)$ satisfies (1) and (2).
First observe that for any $u'\in U'(N)$, we have $C(u')\subset S(N)$. 
Indeed, if there exists $z' \in C(u') \setminus S(N)$,
then there exists $z\in N$ such that $z\leqq z'$.
Since $z'<u'$, we get $z<u'$ contradicting condition (i).
\begin{itemize}
 \item[(1)]
  Let $u'\in U'(N)$. Then we have $C(u')\subset S(N)$ 
  and hence $\bigcup_{u'\in U'(N)} C(u')\subset S(N)$.
  From Claim 1 above, 
  we have $$S(N)=\bigcup_{u\in U(N)} C(u)\subset \bigcup_{u'\in U'(N)} C(u'),$$ 
  and thus $\bigcup_{u'\in U'(N)} C(u')=S(N)$, which proves (1).
 \item[(2)] Now let $u^1,u^2\in U'(N)$. 
 Then we have $C(u^1)\subset S(N)$ and $C(u^2)\subset S(N)$.
  If $C(u^1)\subset C(u^2)$, then (ii) would be violated. This proves (2).
\end{itemize}
\end{proof}

\subsection{Related concepts}\label{sub:otherdef}

In computational geometry, 
\citet{KapRubShaVer08} define \emph{maximal empty orthants} with respect 
to a discrete set of points $Q\subset \R^p$
as partially bounded hyperrectangles of the form 
$\prod_{j=1}^p (-\infty, a_j)\subset \R^p$, 
for some $a\in \R^p$ (apex),
which contain no point of $Q$ and are maximal for this property under inclusion.
It is clear that such points $a$ satisfy the conditions 
of Proposition~\ref{def:upperbounds3}
and are therefore local upper bounds for the search region $S(Q)$.

For their generalization of the two-phase method to MOO problems 
with more than two objectives,
\citet{PrzGanEhr10} are interested in characterizing 
the part of the objective 
space where remaining nondominated points have to be searched after phase one.
To this end, they define a concept similar to our search region, 
the \emph{search area}. 
The search area $S'(N)$ is defined as the closure of the complement set of 
$\{z\in\R^p:\exists z' \in N, z' \leqq z\}$, i.e.,
$$S'(N)={\rm cl}\left(\{z\in\R^p:\exists z' \in N, z' \leqq z\}^C\right).$$

$N$ is defined in \citet{PrzGanEhr10} as 
an \emph{upper bound set for the nondominated set} in the sense of \citet{EhrGan07},
and can also be any stable set of feasible points.
Note that we omit from their definition a lower bounding part, 
which is not relevant for our purpose.

In fact, the search area $S'(N)$ corresponds to the closure of the search region $S(N)$
defined according to equation~(\ref{eq:searchregion}) in Definition~\ref{def:upperbounds1}.
This difference 
implies that the search area $S'(N)$ includes $N$ 
and even points of the objective space that are weakly dominated 
by some points of $N$.

\citet{PrzGanEhr10} and \citet{DaeKla14} also describe the search area by a set of 
\emph{corner points} or \emph{upper bounds} which are the same as the \lub{}s we consider in this {\DocType}.
The former rely on a definition for these points which corresponds to Proposition~\ref{def:upperbounds3}.

\subsection{Application for the solution of MOO problems}

The concepts and properties developed in Sections~\ref{sub:the_search_region} and~\ref{sub:explicit}
apply to MOO in general.
For continuous and mixed discrete-continuous problems, they are useful in approaches
aimed at generating discrete representations of the nondominated set.
In the case of discrete problems, such as multi-objective combinatorial optimization (MOCO) problems,
they play an important role in the generation of the nondominated set as well.
In this section, we focus on the latter issue and 
mention two widely applied methods to show how the computation of
local upper bounds can be integrated into an overall solution strategy.

\paragraph{A generic method based on the solution 
of budget constrained programs} 

The representation of the search region as a set of search zones 
makes it possible to derive a simple algorithm 
to enumerate all nondominated points of 
a MOCO problem. 
This can be done by iteratively exploring the search zones
that define the search region and updating the search region whenever new points are found.
The exploration of a search zone $C(u)$ has to determine 
whether $C(u)$ contains feasible points, and if so output one such point.
In order to limit the number of search zones that are considered, 
the exploration routine should return only nondominated points.
Such an exploration can be achieved by solving, for example, 
the following mathematical program associated to a search zone $C(u)$:
$$P(u) : \min \{g(\Obj(x)):x \in X, \Obj(x) < u\}$$
where $g$ is any strongly increasing aggregation function of the $\Obj_j$'s 
(e.g. $g:z \mapsto \sum_{j=1}^p z_j$).
Note that the strict dominance 
in the definition of $P(u)$ 
can be transformed into non-strict inequalities
by slightly decreasing $u$ 
since $X$ is a discrete set, or, when possible, by taking advantage of Remark~\ref{rem:upperbounds2}.
Problem~$P(u)$ can be seen as a variant of the $\varepsilon$-constraint method \citep[see e.g.][]{ChaHai83}
and was proposed in \citet{ChaLemElz86}
using a weighted sum function with positive weights as function $g$.

The generic method is presented in Algorithm~\ref{alg:basic_moco}. 
From property (2) of Definition~\ref{def:upperbounds1}, there is no redundant
constrained program among the programs 
associated to the search zones of the current search region.

\begin{algorithm}[!htbp]
\SetKwInOut{Input}{input}\SetKwInOut{Output}{output}
\SetKw{SuchThat}{such that}

\Input{$X$, $\Obj$, $\BM$}

\Output{$\ND$}

\BlankLine

$N\leftarrow\emptyset$; $U(N)\leftarrow\{\BM\}$

\While{$U(N)\neq \emptyset$} {
	Select $u\in U(N)$ \nllabel{alg:basic_moco:lub}
	
	\If{$P(u)$ is feasible} {
	
		Let $\bz$ be an optimal point of $P(u)$
	
		$N\leftarrow N \cup \{\bz\}$ 
		
		Update $U(N)$
	}
	\Else {
	
		$U(N)\leftarrow U(N)\setminus\{u\}$
	
	}
}
\Return {$N = \ND$}

\caption{Generic method to generate all nondominated points of a MOCO problem
based on the definition of search zones\label{alg:basic_moco}}
\end{algorithm}

Now we count the number of constrained programs that have to be solved in Algorithm~\ref{alg:basic_moco}.
Note that each \lub{} that is considered at Step~\ref{alg:basic_moco:lub}
will either lead to a nondominated point
or be part of the final \ubs{} $U(\ND)$ 
(if the associated $P(u)$ has no feasible solution).
Therefore the number of calls to the exploration routine is exactly 
$|U(\ND)|+|\ND|$.
This, together with the tight upper bound on $|U(\ND)|$ provided in Section~\ref{sec:counting},
amounts to $O(|\ND|^{\lfloor \frac{p}{2} \rfloor})$ solver calls for $p \geq 2$.

Many papers in the literature propose overall strategies,
based on solving budget constrained programs,
that generate the nondominated set:
\citet{ChaLemElz86, LauThiZit06, SylCre08, OzlAzi09, LokKok12, KirSay14, DaeKla14}.
To our knowledge, only approaches specialized to the bi- and tri-objective cases
provide a non-trivial  upper bound on the number of solver calls.
Good upper bounds are known, however, for $p=2$ and $p=3$.
\citet{ChaLemElz86} propose an approach for the case $p=2$, similar to
Algorithm~\ref{alg:basic_moco},
where exactly $2|\ND|+1$ solver calls are
required.
In the more complex case $p=3$, \citet{DaeKla14} suggest
a closely related method and show that
at most $3|\ND|-2$ solver calls are needed in this case.

\paragraph{MOBB and two phase methods}

In multi-objective branch and bound (MOBB), 
a bounding step is performed at each node of a search tree.
Assume we consider the current node whose set of feasible solutions 
is $X'\subset X$.
In general, computing either $X'$ or $\ObjFeas'=\Obj(X')$ would be expensive.
However, given a set of, say $m$ weight vectors 
$\lambda^1,\dots,\lambda^m$ of $\R^p$ 
such that for any $i\in\{1,\dots,m\}$, $\lambda^i \geq \Bzero$,
we may approximate $\ObjFeas'_{\rm nd}$ by computing 
$\alpha_i=\min\{\sum_{j=1}^p \lambda^i_j  z_j : z \in \ObjFeas'\}$,
especially if the single objective version of the underlying problem 
is solvable in polynomial time.
Denoting by $\ObjSub$ the set 
$\{z\in Z: \sum_{j=1}^p \lambda^i_j  z_j \geq \alpha_i,\, i=1,\dots,m\}$,
we have $\ObjFeas' \subset \ObjSub$, thus $\ObjFeas'_{\rm nd} \subset \ObjSub$.
Therefore, if $\ObjSub \cap S(N) = \emptyset$ then the current node can be pruned 
since it cannot yield any new nondominated point.

Consider also the two-phase method, 
and especially the version where a ranking algorithm is used 
to obtain, in phase two, nondominated non-extreme points.
A set of weight vectors 
$\lambda^1,\dots,\lambda^m$ that satisfy the same conditions as above
is obtained in phase one.
At some time during phase two, we are given $m$ values 
$\alpha_1,\dots,\alpha_m\in\R$
such that for each $i\in\{1,\dots,m\}$, all feasible points whose weighted sum value
according to the weight vector $\lambda^i$ is less than or equal to $\alpha_i$
have been computed.
Considering $N$ as the set of all these feasible points, excluding the dominated ones,
we can test whether the set 
$\ObjSub=\{z\in Z : \sum_{j=1}^p \lambda^i_j  z_j \geq \alpha_i,\,i=1,\dots,m\}$
intersects the search region $S(N)$.

 Now we explain how \lub{}s help to determine whether such a polytope $\ObjSub$
 intersects the search region.
 For any $\lambda \in \R^p$ 
 such that $\lambda \geq \Bzero$ and for any $z, z'\in \R^p$ 
 such that $z<z'$, we have $\sum_{j=1}^p \lambda^i_j  z_j < \sum_{j=1}^p \lambda^i_j  z_j'$. 
 Therefore, together with conditions (1a') and (1b') 
 of Proposition~\ref{def:upperbounds2}, we obtain:
 
\begin{multline}
  \ObjSub\cap S(N) = \emptyset\text{ if, for all } u\in U(N),
  \text{ there exists }i\in\{1, \dots, m\} \\
  \text{ such that }
  \sum_{j=1}^p \lambda^i_j  u_j \leq \alpha_i
  \label{eq:RuleR} \tag{Rule ``\R''}
\end{multline}

 If the feasible points are integral, 
 that is we restrict ourselves to integer vectors in both $S(N)$ and $\ObjSub$,
 the above condition for $\ObjSub\cap S(N)=\emptyset$ can be strengthened using 
 Remark~\ref{rem:upperbounds2}.
 In this case, we rely on the following implication: 
 if $z$ and $z'$ are two vectors such that
 $z\leqq z'$, then $\sum_{j=1}^p \lambda^i_j  z_j \leq \sum_{j=1}^p \lambda^i_j  z_j'$.
 Then in this case, we have:

\begin{multline}
  \ObjSub\cap S(N) = \emptyset \text{ if, for all } u \in U(N),
  \text{ there exists } i\in\{1, \dots, m\}\\
  \text{ such that }
  \sum_{j=1}^p \lambda^i_j  u_j < \alpha_i-\sum_{j=1}^p \lambda^i_j 
  \label{eq:RuleZ} \tag{Rule ``\Z''}
\end{multline}
 
 These rules are used by \citet{SouSpa08a} in the context 
 of MOBB to find all nondominated points
 of the bi-objective minimum spanning tree problem. 
 \citet{PrzGanEhr08} also use them in a two-phase method 
 based on the use of a ranking algorithm to find all nondominated points
 of the bi-objective assignment problem. 
 \citet{PrzGanEhr10} consider their application again in a two-phase 
 method not limited to the bi-objective case. 
 The rules they propose for the general multi-objective case
 are related to their definition 
 of the search area we presented in Section~\ref{sub:otherdef}, 
 which implies that the rule for the integral case 
 is a little weaker than \ref{eq:RuleZ}.

\section{Existence and construction of upper bound sets}\label{sec:existence}

The initial search region consists of the whole search interval $ Z$. 
Therefore, it can be described by the following upper bound set:
$$U(\emptyset) = \{\BM\}.$$
Actually, this defines the unique search zone $C(\BM)= Z$, 
which is consistent with Definition~\ref{def:upperbounds1}.

Starting with this in the case $N=\emptyset$, 
a simple incremental algorithm can be formulated 
that iteratively introduces points to the set $N$ 
and updates the \ubs{} $U(N)$ accordingly.
It was first proposed by \citet{PrzGanEhr10} 
with a slight difference in the filtering step
to which we shall return later.
Given a finite and stable set $N\subset Z$, a corresponding \ubs{} $U(N)$, 
and a point $\bz\in Z$ that is nondominated with respect to $N$,
Algorithm~\ref{alg:basic_incr} describes the updating procedure 
to obtain the \ubs{} $U(N\cup\{\bar{z}\})$. 

\begin{algorithm}[!htbp]
\SetKwInOut{Input}{input}\SetKwInOut{Output}{output}
\SetKw{SuchThat}{such that}

\Input{$U(N)$, $\bz$\tcp*{Set of \lub{}s and new point}}

\Output{$U(N\cup \{\bz\})$}

\BlankLine

$A\leftarrow\{u\in U(N):\bz<u\}$
\nllabel{alg:basic_incr:visible} 
\tcp*{Search zones that contain $\bz$}

{$B\leftarrow\{u\in U(N)\setminus A:\bz \leq u\}$}
\tcp*{Search zones whose boundary contains $\bz$}

$P\leftarrow\emptyset$

\For{$u\in A$}{\nllabel{alg:basic_incr:start_projections}

\For{$j\in\onetod$}{

$P\leftarrow P\cup \{(\bz_{j},u_{-j})\}$
\nllabel{alg:basic_incr:projections}
\tcp*{Generate all projections of $\bz$ on the \lub{}s of $A$}
}}

$P\leftarrow\{(\bz_j,u_{-j})\in P:(\bz_j,u_{-j})\not\leq u',\,
\forall u' \in P\cup B \}$
\nllabel{alg:basic_incr:filtering} 
\tcp*{Filter out all {redundant points} of $P$}

$U(N\cup \{\bz\})\leftarrow (U(N)\setminus A) \cup P$
\nllabel{alg:basic_incr:final_step}

\caption{Update procedure of an \ubs{} 
based on redundancy elimination
\label{alg:basic_incr}}
\end{algorithm}

Basically, Algorithm~\ref{alg:basic_incr} updates each search zone $C(u)$ 
in which the new point $\bar{z}$ lies by removing 
from $C(u)$ the part of $ Z$ 
which is dominated by $\bar{z}$ (including $\bar{z}$). 
This is achieved by replacing $C(u)$ by $p$ subzones as done 
in Step~\ref{alg:basic_incr:projections}
of Algorithm \ref{alg:basic_incr}. 
Some of these newly generated subzones may be redundant, 
and are thus removed in Step~\ref{alg:basic_incr:filtering}.
More formally, we state the following result, which justifies Algorithm~\ref{alg:basic_incr}.

\begin{proposition}\label{prop:existence}
Let $N\subset Z$ be a non-empty finite and stable set of points.
Applying Algorithm~\ref{alg:basic_incr} iteratively on the points of $N$, 
starting with
an initial \ubs{} $U(\emptyset)=\{\BM\}$,
returns the correct upper bound set $U(N)$.
\end{proposition}

\begin{proof}
 We show that Algorithm~\ref{alg:basic_incr} correctly computes 
 the set $U(N\cup\{\bar{z}\})$, given any finite and stable set $N\subset Z$ 
 of points, 
 the correct upper bound set $U(N)$, and a new point $\bar{z}\in Z$ 
 that is nondominated with respect to $N$. The result then follows by induction.
 
 Considering the new point $\bar{z}$, the search region $S(N\cup\{\bar{z}\})$ 
 must be updated from $S(N)$ by removing all points in $Z$ such that 
 $\bar{z}\leqq z$.

 In Step~\ref{alg:basic_incr:visible} of Algorithm \ref{alg:basic_incr} 
 the search zones $C(u)$, $u\in A$, containing $\bar{z}$ 
 are identified. All other search zones $C(u)$, $u\in U(N)\setminus A$, 
 are not affected by the new point $\bar{z}$ and thus need not be modified.
 
 Thus, we just need to remove the set of points 
 $\{z\in S(N) : \bar{z} \leqq z\}$ from the search zones 
 $C(u)$, $u\in A$, to ensure that condition (1) of Definition~\ref{def:upperbounds1} is
 satisfied. Steps~\ref{alg:basic_incr:start_projections}-\ref{alg:basic_incr:projections}
 are justified by the fact
 that for any $u\in A$ we have
 $$C(u) \setminus \{z\in S(N) : \bar{z} \leqq z\} = \bigcup_{j=1}^p C(\bar{z}_j,u_{-j}).
 $$

 Among the candidate \lub{}s of $P$, 
 as defined after all iterations of Step~\ref{alg:basic_incr:projections}, 
 there may be some \emph{redundant} \lub{}s 
 in the sense
 that they induce search zones that are included 
 in a search zone associated to some (candidate) \lub{}s 
 of $P\cup (U(N) \setminus A)$.
 Let $(\bz_j, u_{-j})\in P$, with $u\in A$, be a redundant \lub{},
 i.e.\ there exists $u'\in P\cup (U(N) \setminus A)$
 such that $\bz \leq (\bz_j, u_{-j}) \leq u'$.
 If $\bz<u'$, then $P$ contains the candidate \lub{} $(\bz_j, u'_{-j})$, 
 otherwise we have $u'\in B$.
 Therefore, Step~\ref{alg:basic_incr:filtering} correctly filters the set $P$,
 which leads to satisfying condition (2) of Definition~\ref{def:upperbounds1}. 
\end{proof}

In \citet{PrzGanEhr10}, the filtering step is formulated
with respect to the set $U(N)$, i.e. 
$$P\leftarrow\{(\bz_j,u_{-j})\in P:
(\bz_j,u_{-j})\not\leq (u'),\, \forall u'\in U(N)\}.$$
This is correct, but involves unnecessary dominance tests 
compared to Algorithm~\ref{alg:basic_incr}, 
since one only needs to filter with respect to $P\cup B$
instead of $U(N)$.

It is even possible to further refine the filtering step 
of Algorithm~\ref{alg:basic_incr}.
To this end, we prove the following proposition.
\begin{proposition}\label{rem:basic}
Let $(\bz_j, u_{-j})$ be a candidate \lub{} in $P$. Then:
\begin{itemize}
 \item[(1)] $(\bz_j, u_{-j})\leq (\bz_k,u'_{-k})$ for some 
 $(\bz_k,u'_{-k})\in P$ with $k\neq j$ 
 implies $(\bz_j, u_{-j})\leq (\bz_j,u'_{-j})$;
 
 \item[(2)] $(\bz_j, u_{-j})\leq u'$ for some $u'\in B$ 
 implies $\bz_j = u_j'$ and $\bz_{-j} < u_{-j}'$.
\end{itemize}
\end{proposition}

\begin{proof} \hfill
\begin{itemize}
 \item[(1)] Since $\bz < u'$, we have $(\bz_k,u'_{-k})\leq u'$, 
 which, together with $(\bz_j, u_{-j})\leq (\bz_k,u'_{-k})$,
 leads to $(\bz_j, u_{-j})\leq u'$, 
 and thus $(\bz_j, u_{-j})\leq (\bz_j,u'_{-j})$.
 \item[(2)] Since $\bz < u$ and $(\bz_j, u_{-j})\leq u'$ 
 we have $\bz_{-j} < u_{-j} \leq u'_{-j}$. Moreover, with $u'\in B$, 
 we obtain $\bz_j=u_j'$.
\end{itemize}
\end{proof}

According to property (1) of Proposition~\ref{rem:basic}, 
the filtering step~\ref{alg:basic_incr:filtering}
of Algorithm~\ref{alg:basic_incr} can be replaced by the following step:
\begin{equation*}
\begin{split}
P\leftarrow\{(\bz_j,u_{-j})\in P:
(\bz_j,u_{-j})\not\leq (\bz_j,u_{-j}'),\\ 
\forall (\bz_j,u_{-j}') \in P
\text{ and }(\bz_j,u_{-j})\not\leq u',\, &\forall u'\in B\}   
  \end{split} 
\end{equation*}
which is equivalent to the following formulation:
$$P\leftarrow\{(\bz_j,u_{-j})\in P:
(\bz_j,u_{-j})\not\leq u',\, \forall u'\in (A\cup B)\setminus\{u\}\}$$

From property (2) of Proposition~\ref{rem:basic},
it is also possible to do fewer dominance tests against the \lub{}s
of $B$.

Overall, Proposition~\ref{rem:basic} shows that it is only required to
perform dominance tests between vectors that differ 
in all but one component.
We present these enhancements in Algorithm~\ref{alg:basic_incr2}, 
where we split the sets $B$ and $P$ into $p$ disjoint sets, respectively
$B_1,\dots,B_p$ and $P_1,\dots,P_p$,
to stress the by-component filtering step.

\begin{algorithm}[!htbp]
\SetKwInOut{Input}{input}\SetKwInOut{Output}{output}
\SetKw{SuchThat}{such that}

\Input{$U(N)$, $\bz$\tcp*{Set of \lub{}s and new point}}

\Output{$U(N\cup \{\bz\})$}

\BlankLine

$A\leftarrow\{u\in U(N):\bz<u\}$
\tcp*{Search zones that contain $\bz$}

\For{$j\in\onetod$}{
	{$B_j\leftarrow\{u\in U(N):\bz_j = u_j\text{ and }\bz_{-j} < u_{-j}\}$}
	
	$P_j\leftarrow\emptyset$
}

\For{$u\in A$}{\nllabel{alg:basic_incr2:start}

\For{$j\in\onetod$}{

$P_j\leftarrow P_j\cup \{(\bz_{j},u_{-j})\}$
\nllabel{alg:basic_incr2:projections}
\tcp*{Generate all projections of $\bz$ on the \lub{}s of $A$}
}}

\For{$j\in\onetod$}{\nllabel{alg:basic_incr2:start_filtering} 
$P_j\leftarrow\{(\bz_j,u_{-j})\in P_j:(\bz_j,u_{-j})\not\leq u',\,
\forall u' \in P_j\cup B_j \}$
\nllabel{alg:basic_incr2:filtering} 
\tcp*{Filter out all {redundant points} of $P$}
}

$U(N\cup \{\bz\})\leftarrow (U(N)\setminus A) \cup \bigcup_{j=1}^p P_j$

\caption{Update procedure of an \ubs{} 
based on redundancy elimination
with an enhanced filtering step
\label{alg:basic_incr2}}
\end{algorithm}

While Algorithm~\ref{alg:basic_incr2} allows the correct computation 
of \ubs{}s, 
it requires the iterative filtering 
for a possibly large number of candidate \lub{}s,
which may be computationally expensive.
In the next section, we establish structural properties of \lub{}s 
which yield necessary and sufficient conditions for a candidate \lub{}
to become actually a (non-redundant) \lub{}.
Then a new approach to the incremental computation of an \ubs{},
which avoids the filtering step, is derived.

\section{Properties of local upper bounds and their efficient computation}
\label{sec:properties}

In this section, we study some theoretical properties of \lub{}s 
that yield another approach which, in comparison to the algorithms
presented in Section~\ref{sec:existence}, 
avoids the filtering step 
(namely Steps~\ref{alg:basic_incr2:start_filtering} and \ref{alg:basic_incr2:filtering} 
in Algorithm~\ref{alg:basic_incr2}).

The properties are first presented under a simplifying assumption that 
no two distinct points,
among the points of $Z$ to be considered,
share the same value in any dimension.
This assumption, denoted ``\GP{}'' in the remainder, 
corresponds to what is referred to as a \emph{general position} assumption 
in computational geometry.
It is, however, not realistic for many instances of MOCO problems, 
that is why we extend the properties under the general case
according to which identical component values among distinct points are allowed.

We first illustrate the properties on small examples
(Section~\ref{sec:properties_ex}). Then we detail the properties and derive 
the new approach (Section~\ref{sec:incremental_gp}).

\subsection{Introductory examples and geometric interpretation}
\label{sec:properties_ex}

In this section, we give a geometric intuition 
to the properties that are detailed in the next sections
through two example instances in the tri-objective case.
First we present an example instance in the \GP{} case (Example~\ref{ex:gp}).
Then we discuss the consequences of feasible points having identical component values 
(Example~\ref{ex:ngp}).

\begin{example}[Under \GP{}]\label{ex:gp}
We consider a three-dimensional simple instance 
of our problem which consists of 
two feasible points: $z^1 = (3,5,7)$ and $z^2 = (6,2,4)$. 
Let us apply the incremental algorithm 
presented in Section~\ref{sec:existence}
first on $U(\emptyset)=\{\BM\}$ 
and $\bz=z^1$ 
then on $U(\{z^1\})$ and $\bz=z^2$.
At the first iteration, $z^1$ yields three \lub{}s,
namely $u^1=(3,M,M)$, $u^2=(M,5,M)$,
and $u^3=(M,M,7)$ so that $U(\{z^1\})=\{u^1, u^2, u^3\}$.
Then at the second iteration we consider the three projections of $z^2$ 
on the \lub{}s whose associated search zones contain $z^2$ 
which are $u^2$ and $u^3$. 
We get $u^{21} = (6,5,M)$, $u^{22} = (M,2,M)$, and $u^{23} = (M,5,4)$ for $u^2$,
and $u^{31} = (6,M,7)$, $u^{32} = (M,2,7)$, and $u^{33} = (M,M,4)$ for $u^3$.
Projections $u^{23}$ and $u^{32}$ being redundant since
$u^{33}\geq u^{23}$ and $u^{22} \geq u^{32}$, we have 
$U(\{z^1, z^2\}) = \{u^1,u^{21}, u^{22}, u^{31}, u^{33}\}$.

We represent the situation in Figure~\ref{fig:3d_gp1}. 
The feasible points $z^1,z^2$ are depicted 
together with their Pareto dominance cones 
$\{z\in Z: z^i \leqq z\}$, $i\in\{1, 2\}$, in gray
as well as the \lub{}s. 
The scene is represented in perspective from point $\Bm$ to point $\BM$
so that the search zones go towards us.

\begin{figure}[!htbp]
\begin{centering}
\subfloat[$N=\{z^1\}$\label{fig:3d_gp1_1}]{
\includegraphics[width=.48\textwidth]{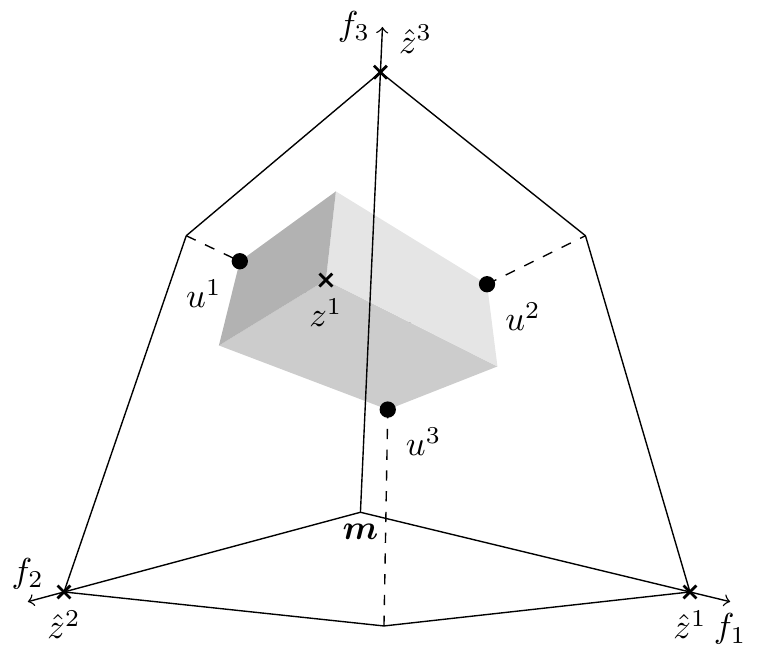}}
\subfloat[$N=\{z^1,z^2\}$\label{fig:3d_gp1_2}]{
\includegraphics[width=.48\textwidth]{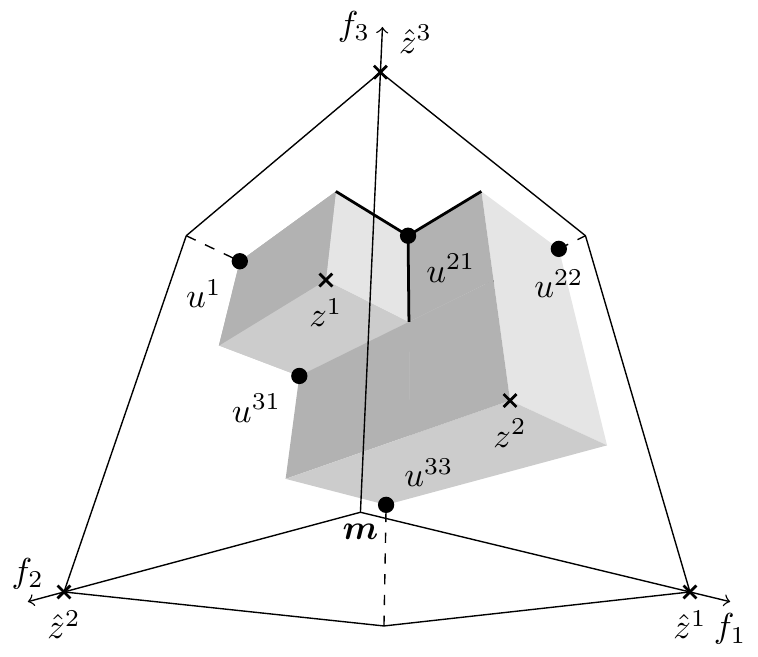}}
\par\end{centering}
\caption{A three-dimensional example problem with points under \GP{}
\label{fig:3d_gp1}}
\end{figure}

Now we look at a particular \lub{}, say $u^{21}=(z^2_1, u^2_{-1})$. 
Consider any point $\bz$ that belongs to the search zone defined by $u^{21}$.
The $j$th projection of $\bz$ 
on $u^{21}$ amounts to sliding $u^{21}$ 
along the half-line $[u^{21},(\Bm_j,u^{21}_{-j}))$.
From Figure~\ref{fig:3d_gp1}, we can see that 
if a projection of $\bz$ on $u^{21}$
lies outside any of the three black line segments that start from $u^{21}$,
then it will be redundant since it belongs to the closure of another search zone.
We can see that these line segments are edges of the union of dominance cones 
associated to the points of $N$, plus three dummy points
$\dummy^1=(M,m,m)$, $\dummy^2=(m,M,m)$ and $\dummy^3=(m,m,M)$.
With these dummy points, even \lub{}s located on the boundary
of $\hZ$ lie at the intersection of three dominance cones.
We can now avoid the filtering step 
(Steps~\ref{alg:basic_incr2:start_filtering} and \ref{alg:basic_incr2:filtering}) 
of Algorithm~\ref{alg:basic_incr2} if,
for each \lub{} $u$, the $p$ edges of the union 
which are incident to $u$ are known.
In the rest of this section, we consider facets of the union of
the dominance cones associated to the points of $N$ 
and to the dummy points $\dummy^1, \dummy^2, \dummy^3$.

We can see that the facets incident to $u^{21}$ 
are composed of two facets incident to $u^2$ that are shrunk 
after the first projection of $z^2$ and one facet 
which is a subset of a facet of the dominance cone associated to $z^2$.
So, in order to compute the edges incident to $u^{21}$, 
we only have to keep track of the three points 
that lower bound the facets, namely $z^2$, $z^1$, and $\dummy^3$. 
This holds because under \GP{}, a facet is defined by a \lub{} 
and a single point of $N$.
\end{example}

\begin{example}\label{ex:ngp}

\begin{figure}[!htbp]
\begin{centering}
\includegraphics{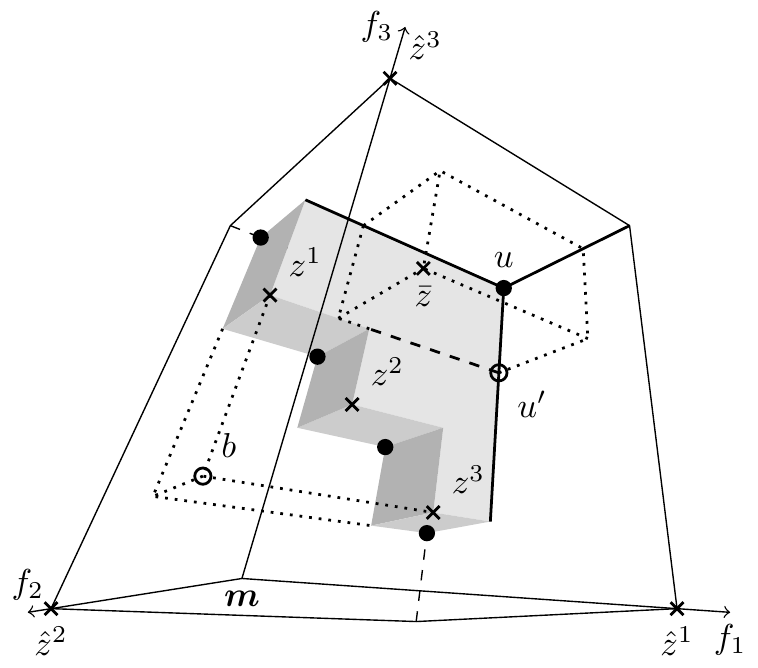}
\par\end{centering}
\caption{A three-dimensional instance with feasible points having the same value on component~2
\label{fig:3d_ngp1}}
\end{figure}

Consider the three-dimensional instance represented in Figure~\ref{fig:3d_ngp1}
with three feasible points $z^1=(2,7,7)$, $z^2=(5,7,5)$ and $z^3=(8,7,3)$,
which all share the same value on the second coordinate.
We look again at facets of the union of all dominance cones associated to the points of $N$.
The \lub{} $u=(M,7,M)$ is defined by $\dummy^1$ on component 1, $\dummy^3$ on component 3, 
and $z^1$, $z^2$, and $z^3$ on component 2.
We consider the facet of the union incident to $u$ and orthogonal to the $\Obj_2$-axis. 
Similarly to the \GP{} case, we may want to represent this facet by $u$ 
and a single point defining a lower bound on the $\Obj_1$ and $\Obj_3$ values.
Since this facet is incident to three feasible points
we could define $b=(z^1_1, z^1_2 = z^3_2, z^3_3) = (2,7,3)$ (see again Figure~\ref{fig:3d_ngp1}).

However, this information may not be sufficient to avoid future redundancies.
Consider for example the point $\bz = (4,3,7)$ 
as depicted in Figure~\ref{fig:3d_ngp1}
together with its Pareto dominance cone (in dotted lines).
It satisfies $z^1_1<\bz_1<z^2_1$, $\bz_2 < z^1_2=z^2_2$ and $\bz_3 = z^1_3$.
$\bz$ defines among others the \lub{} $u'=(\dummy^1_1,z^2_2, \bz_3) = (M,7,7)$.
Unfortunately, one of the edges incident to $u'$ represented 
as a dashed line in the figure is limited by an intermediate feasible point, namely $z^2$.
Therefore, it will be necessary in the general case to keep track 
of all feasible and dummy points that belong to a facet incident to a \lub{}.
This is what the sets $Z^j(\cdot)$ are aimed at in Section~\ref{sec:incremental_gp}.
\end{example}

\subsection{Theoretical properties of local upper bounds and a new incremental approach}
\label{sec:incremental_gp}

According to Step~\ref{alg:basic_incr2:projections} 
of Algorithm~\ref{alg:basic_incr2}, 
all components of a \lub{} $u$ result from previously generated upper bounds 
for $p-1$ components 
and, for the remaining component, from the currently added point $\bar{z}$.
The initial \lub{} $\BM$, however, is not defined from any point of $Z$.
In order to make no particular case of the component values 
inherited from $\BM$,
we extend any stable set of points from $Z$ with the dummy points we introduced
in the previous section.
Namely, we define the \emph{extension} of $N$ as the set 
$\hN=N \cup \{\dummy^j,\,j=1,\dots,p\}$, where
$$\dummy^j = (\BM_j,\Bm_{-j}),\,j\in\onetod$$
It is not hard to see that $U(\{\dummy^j,\,j=1,\dots,p\})=\{\BM\}$, i.e.
the dummy points yield the initial search zone,
which implies that for any finite and stable set $N$ of points from $Z$,
we have $U(\hN)=U(N)$.

Using dummy points, we now have that any component value 
of a local upper bound is defined by a point of $\hN$.

Observe that a dummy point $\dummy^j$ can only define the $j$th component
of any local upper bound, which is $M$. Indeed since no point from $Z$ is lower
than or equal to $m$ on any component, 
$m$ cannot be a component value of a local upper bound.
Therefore, and since $M$ is unique in the component values of a dummy point,
$\dummy^j$ is the only dummy point that can define component $j$.

The following proposition gives a useful property 
of those points that define each component of a \lub{}.

\begin{proposition}\label{lemma:ab}
 For any \lub{} $u\in U(N)$ and $j \in \onetod$, 
 there exists $z\in \hN$ such that $z_j=u_j$ and $z_{-j}<u_{-j}$.
\end{proposition}

\begin{proof}
 If $u_j=M$, then the dummy point $\dummy^j\in \hN$ satisfies the required conditions.
 
 Otherwise and since $\hN$ is a finite set, there exists an $\varepsilon>0$
 sufficiently small such that no point of $\hN$ 
 has its $j$th component value in the interval $(u_j,u_j+\varepsilon)$.
 Let $u'=(u_j +\varepsilon, u_{-j})$. 
 According to Proposition~\ref{def:upperbounds3}, since $u'\in\hZ$
 and $u'\geq u$, there exists a $z\in N$ such that
 (i)  $z<u'$ and
 (ii) $z\not<u$.
 It follows from (i) that we have $z_{-j}<u_{-j}$,
 which imposes $z_j\geq u_j$ from (ii).
 From the choice of $\varepsilon$, we therefore have $z_j=u_j$.
\end{proof}

In the following we define two notations for those points that define local upper
bounds, for the general case and for the \GP{} case, respectively.

\begin{definition}
 For any local upper bound $u\in U(N)$, we denote by 
 $Z^j(u)=\{z\in\hN:z_j=u_j \text{ and } z_{-j}<u_{-j}\}$
 the set of defining points of $u$ for component~$j$, $j\in\onetod$.
 
 In the \GP{} case, the unique defining point of $u$ for component $j$ 
 is denoted $z^j(u)$.
\end{definition}

Using Proposition~\ref{lemma:ab},
we can now precisely characterize the projections 
which are kept in the set $P$ 
after the filtering step of Algorithm~\ref{alg:basic_incr2}.
We first consider the \GP{} case.

\begin{theorem}[\GPl{}]\label{th:incremental}
 Let $\bz$ be a point of $Z$ that is nondominated with respect to $N$ 
 and such that the points in $N\cup\{\bz\}$ satisfy \GP{}. 
 Consider a \lub{} $u\in U(N)$ such that $\bz < u$. 
 Let $z_j^{\max}(u)=\max_{k\neq j} \{z_j^{k}(u)\}$.
 
 Then, for any $j\in\{1,\dots,p\}$, $({\bz}_j,u_{-j})$ is a \lub{} of $U({N}\cup \{\bz\})$ 
 if and only if  $\bz_j>z_j^{\max}(u)$.
\end{theorem}

\begin{proof} Let $u\in U(N)$ and $\bz\in Z\setminus N$ be a point 
not dominated by any point of $N$ such that $\bz<u$.
\begin{itemize}
 \item[$(\Rightarrow)$] 
  Suppose that $\bu=(\bz_j,u_{-j})$ is a \lub{} in $U(N\cup\{\bz\})$ 
  and let $z_j^{\max}(u)=z_j^{k}(u)$ 
  for some point $z^{k}(u)\in \hN$, such that $z^{k}_k(u)=u_k$, $k\neq j$.
  Therefore, $z^{k}_k(u)=\bu_k$ and, from \GP{},
  no other point of $\hN$ equals  $\bu_k$
  on its $k$th component.
  Thus from Proposition~\ref{lemma:ab}, we have $z^{k}_{-k}(u) < \bu_{-k}$,
  which implies $z^k_j(u)=z_j^{\max}(u) < \bu_j = \bz_j$.

 \item[$(\Leftarrow)$] 
  Assume that for a given $j\in \onetod$, ${\bz}_j > z^{\max}_j(u)$.
  Suppose, to the contrary, that $(\bz_j,u_{-j})$ 
  is not a \lub{} for ${N}\cup\{\bz\}$, 
  that is,
  it dominates a \lub{} of $U(N\cup\{\bz\})$.
	Hence from Proposition~\ref{rem:basic}, there exists $u'\in U(N)$ 
	such that $(\bz_j,u_{-j})\leq(\bz_j,u'_{-j})$
	(note that in the \GP{} case, the set $B$ defined in Algorithm~\ref{alg:basic_incr} is empty).
  Then, we have $u_{-j} \leq u'_{-j}$, which implies $u_j>u'_j$
  and $u_k<u'_k$ for some $k\neq j$. 
  Let $z^{k}(u)\in {\hN}$ be the point 
  that defines the $k$th component $u_k$ of $u$. 
  From Proposition~\ref{lemma:ab}, we have $z^{k}_{-k}(u)<u_{-k}$. 
  Thus, since $k\neq j$, we have $z^{k}_{-j}(u)<u'_{-j}$ but since $u'$ is a \lub{}, 
  we must have $z^{k}_{j}(u)\geq u'_{j}$ (otherwise $z^{k}(u)<u'$). 
  Hence, $z_j^{\max}(u) \geq z^{k}_{j}(u) \geq u'_j$. 
  Since we have both $\bz_j>z_j^{\max}(u)$ and $\bz_j<u'_j$, 
  we obtain a contradiction:
  $\bz_j < u'_j \leq z^{\max}_j(u) < \bz_j$.
 \end{itemize}
\end{proof}

Let us illustrate this theorem on the first example instance of Section~\ref{sec:properties_ex}.

\setcounter{example}{1}
\begin{example}[continued] 
Consider the situation in Figure~\ref{fig:3d_gp1_1} with $N=\{z^1\}$,
where $z^1=(3,5,7)$.
The points that define the \lub{}s of $U(N)$, namely $u^1=(3,M,M)$, $u^2=(M,5,M)$, 
and $u^3=(M,M,7)$, are:
$$\begin{array}{lll}
	z^1(u^1) = z^1 & z^2(u^1) = \dummy^2 & z^3(u^1) = \dummy^3 \\
	z^1(u^2) = \dummy^1 & z^2(u^2) = z^1 & z^3(u^2) = \dummy^3 \\
	z^1(u^3) = \dummy^1 & z^2(u^3) = \dummy^2 & z^3(u^3) = z^1 \\
  \end{array}
$$
and $z^{\max}(u^1)=(m,5,7)$, $z^{\max}(u^2)=(3,m,7)$, 
and $z^{\max}(u^3)=(3,5,m)$.

The point $z^2=(6,2,4)$ strictly dominates $u^2$ and $u^3$ and we have:
$$\begin{array}{lll}
	z^2_1 > z^{\max}_1(u^2) & z^2_2 > z^{\max}_2(u^2) & z^2_3 \leq z^{\max}_3(u^2) \\
	z^2_1 > z^{\max}_1(u^3) & z^2_2 \leq z^{\max}_2(u^3) & z^2_3 > z^{\max}_3(u^3) \\
  \end{array}
$$
thus we obtain again the four new \lub{}s $u^{21}=(z^2_1,u^2_{-1})$,
$u^{22}=(z^2_2,u^2_{-2})$,
$u^{31}=(z^2_1,u^3_{-1})$, and
$u^{33}=(z^2_3,u^3_{-3})$.
\end{example}

According to Theorem~\ref{th:incremental}, 
we can avoid the filtering step of Algorithm~\ref{alg:basic_incr2} 
if we keep track of the $p$ points that define each \lub{} 
and only generate the projections of $\bz$ 
that satisfy the conditions of Theorem \ref{th:incremental}. 
The corresponding algorithm is detailed in Algorithm~\ref{alg:enh_incr}.
 
\begin{algorithm}[!htbp]
\SetKwInOut{Input}{input}\SetKwInOut{Output}{output}
\SetKw{SuchThat}{such that}

\Input{$U(N)$ together with $z^j(u),\,\forall j\in\onetod,u\in U(N)$
\tcp*{Set of local upper bounds and associated defining points}}

\Input{$\bz$\tcp*{New point}}

\Output{$U(N\cup \{\bz\})$}

\BlankLine

$A\leftarrow\{u\in U(N):\bz<u\}$ \tcp*{Search zones that contain $\bz$}
$P\leftarrow\emptyset$

\For{$u\in A$}{\nllabel{alg:enh_incr:first_step}

	\For{$j\in\onetod$}{
		$z_j^{\max}(u) \leftarrow \max_{k\neq j} \{z_j^{k}\}$
		\nllabel{alg:enh_incr:zmax}
		
		\tcp{Check for the condition of Theorem~\ref{th:incremental}}
		\If{$\bz_j>z^{\max}_j(u)$}
		{\nllabel{alg:enh_incr:begin}
			\tcp{Let $u^j = (\bz_{j},u_{-j})$}
			
			$P\leftarrow P \cup \{u^j\}$
			
			$z^j(u^j) \leftarrow \bz$\nllabel{alg:enh_incr:keep_z}
			
			\For {$k \in \onetod\setminus\{j\}$} {
				$z^k(u^j)\leftarrow z^k(u)$\nllabel{alg:enh_incr:final_step}
			}
		}
	}
}

$U(N\cup \{\bz\})\leftarrow (U(N)\setminus A) \cup P$

\caption{Update procedure of an \ubs{} based on the avoidance of redundancies:
\emph{\GP{} case}\label{alg:enh_incr}}
\end{algorithm}

Note that each component of the vector $z^{\max}(u)$ for a given \lub{} $u$
will be used at most once in all iterations of Algorithm~\ref{alg:enh_incr}. 
That is why it is computed only before its use, 
namely at Step~\ref{alg:enh_incr:zmax}.
Moreover, this vector is not sufficient to compute the vector 
$z^{\max}(u^j)$ associated to a \lub{} $u^j$ defined from $u$.
Indeed, as the following example shows, it is required to keep track 
of all points that define the component values of $u^j$,
as is done in Steps~\ref{alg:enh_incr:keep_z}-\ref{alg:enh_incr:final_step}.

\setcounter{example}{1}
\begin{example}[continued] 
Consider a new point $z^3 = (4,4,2)$ and the \lub{}
$u^{21} = (6,5,M)$ with $z^{\max}(u^{21}) = (3,2,7)$, 
stemming from 
$z^1(u^{21}) = z^2 = (6,2,4)$,
$z^2(u^{21}) = z^1 = (3,5,7)$, and
$z^3(u^{21}) = \dummy^3 = (m,m,M)$.

We have $z^3 < u^{21}$ and $z^3_2 \geq z^{\max}_2(u^{21})$, 
thus $u^{212} = (z^3_2, u^{21}_{-2})$ is a \lub{} of $U(\{z^1, z^2, z^3\})$.
Since 
$z^1(u^{212}) = z^1(u^{21})$, 
$z^2(u^{212}) = z^3$, and
$z^3(u^{212}) = z^3(u^{21})$,
we have $z^{\max}(u^{212}) = (4,2,4)$. 
As we can see, the last component value of $z^{\max}(u^{212})$,
which comes from $z^1(u^{212}) = z^2$,
cannot be obtained from $z^{\max}(u^{21})$ or $z^3$.
\end{example}

In the general case,
Theorem~\ref{th:incremental} is modified as follows: 

\begin{theorem}\label{th:incremental_ngp}
 Let ${N}$ be a finite and stable set of points of $ Z$, 
 and let $\bz$ be a point of $ Z$ that is nondominated with respect to $N$. 
 Consider a \lub{} $u\in U({N})$ such that $\bz < u$. 
 Let $z_j^{\max}(u)=\max_{k\neq j} \min \{z_j : z \in Z^k (u)\}$.
 
 Then, for any $j\in\{1,\dots,p\}$, $({\bz}_j,u_{-j})$ is a \lub{} of $U({N}\cup \{\bz\})$ 
 if and only if $\bz_j>z_j^{\max}(u)$.
\end{theorem}

\begin{proof} Let $u\in U(N)$ and $\bz\in Z\setminus N$ be a point 
not dominated by any point of $N$ such that $\bz<u$.
\begin{itemize}
 \item[$(\Rightarrow)$] 
  Suppose that $\bu=(\bz_j,u_{-j})$ is a \lub{} in $U(N\cup\{\bz\})$ 
  and to the contrary $\bz_j \leq z_j^{\max}(u)$.
  Then, there is $k \neq j$ 
  such that $\bz_j \leq z_j$ for all $z\in Z^k(u)$.
  Since $\bu \leq u$ and $\bu_k = u_k$, 
  it holds that $Z^k(\bu) \subset Z^k(u)$ 
  but for any $z \in Z^k(u)$, $z_j \not< \bz_j = \bu_j$.
  Hence, $Z^k(\bu) = \emptyset$ which contradicts 
  Proposition~\ref{lemma:ab}. 

 \item[$(\Leftarrow)$] 
  Assume that for a given $j\in \onetod$, ${\bz}_j > z^{\max}_j(u)$.
  Suppose, to the contrary, that $(\bz_j,u_{-j})$ 
  is not a \lub{} for ${N}\cup\{\bz\}$, 
  that is,
  it dominates a \lub{} of $U(N\cup\{\bz\})$.
	Hence from Proposition~\ref{rem:basic}, there exists $u'\in U(N)$,
	$(\bz_j,u_{-j})\leq(\bz_j,u'_{-j})$ (possibly with $\bz_j = u'_j$).

  Then we have $u_{-j} \leq u'_{-j}$ which implies $u_j>u'_j$ 
  and $u_k<u'_k$ for some $k\neq j$. 
  From Proposition~\ref{lemma:ab}, the set $Z^k(u)$ is non-empty.
  For any $z\in Z^k(u)$, we have $z_{-k}<u_{-k}$ 
  and thus, since $k \neq j$, $z_{-j}<u'_{-j}$ but since $u'$ is a \lub{}, 
  we must have $z_{j}\geq u'_{j}$ (otherwise $z<u'$). 
  Hence, there is $z\in Z^k(u)$ such that $z_j^{\max}(u) \geq z_{j} \geq u'_j$. 
  Since we have both $\bz_j>z_j^{\max}(u)$ and $\bz_j \leq u'_j$, 
  we obtain a contradiction:
  $\bz_j \leq u'_j \leq z^{\max}_j(u) < \bz_j$.
 \end{itemize}
\end{proof}

We illustrate the general case on the second example instance 
of Section~\ref{sec:properties_ex}.

\begin{example}[continued] 
In Figure~\ref{fig:3d_ngp1}, we consider the situation 
with $N = \{z^1, z^2, z^3\}$ where $z^1=(2,7,7)$, $z^2=(5,7,5)$ 
and $z^3=(8,7,3)$.
We only look at the \lub{} $u=(M,7,M)$.
We have $Z^1(u)=\{\dummy^1\}$, $Z^2(u)=\{z^1, z^2, z^3\}$, and
$Z^3(u)=\{\dummy^3\}$. Thus $z^{\max}(u)=(2,m,3)$. 
The projections of a point $\bz$ that strictly dominates $u$ 
will be kept as non-redundant \lub{}s depending on the comparisons 
between the component values of $\bz$ and $z^{\max}(u)$ only.
\end{example}

Algorithm~\ref{alg:enh_incr_ngp} presents 
the update procedure in the general case.
The initialization is done with $U(\emptyset)=\{\BM\}$ and $Z^j(\BM)=\{\dummy^j\}$,
$j\in\onetod $.

 \begin{algorithm}
\SetKwInOut{Input}{input}\SetKwInOut{Output}{output}
\SetKw{SuchThat}{such that}

\Input{$U(N)$ together with $Z^j(u)$ for all $j\in\onetod$, 
$u\in U(N)$
\tcp*{Set of local upper bounds and associated defining points}}

\Input{$\bz$\tcp*{New point}}

\Output{$U(N\cup \{\bz\})$}

\BlankLine

$A\leftarrow\{u\in U(N):\bz<u\}$ \tcp*{Search zones that contain $\bz$}

$P\leftarrow\emptyset$

\tcp{Update sets $Z^j (u)$ when $\bz$ satisfies the conditions 
of Proposition~\ref{lemma:ab}}

\For{$u\in U(N)$ and $j\in\onetod$ \SuchThat $\bz_j=u_j$ 
and $\bz_{-j} < u_{-j}$}{
	$Z^j(u) \leftarrow Z^j(u) \cup \{\bz\}$	
}

\For{$u\in A$}{
	\For{$j\in\onetod$}{
		$z_j^{\max}(u)\leftarrow \max_{k\neq j} \min \{z_j : z \in Z^k (u)\}$
		
		\tcp{Check for the condition of Theorem~\ref{th:incremental_ngp}}
		
		\If{$\bz_j>z^{\max}_j(u)$} 
		{
			\tcp{Let $u^j = (\bz_{j},u_{-j})$}
			
			$P\leftarrow P \cup \{u^j\}$
			
			$Z^j(u^j) \leftarrow \{\bz\}$
			
			\For {$k \in \onetod\setminus\{j\}$} {
				$Z^k(u^j)\leftarrow \{z\in Z^k(u) : z_j < \bz_j\}$\nllabel{alg:enh_incr_ngp:Zk}
			}
		}
	}
}

$U(N\cup \{\bz\})\leftarrow (U(N)\setminus A) \cup P$

\caption{Update procedure of an \ubs{} based on the avoidance of redundancies:
\emph{general case}\label{alg:enh_incr_ngp}}
\end{algorithm}

\section{Complexity and computational experiments}
\label{sec:complexity_experiments}

In Sections~\ref{sec:existence} and \ref{sec:properties}, we described 
two incremental approaches for the update of an \ubs{}.
In Section~\ref{sec:existence}, the approach is based on
redundancy elimination (RE) among \lub{}s, 
while in Section~\ref{sec:properties} it is based on 
redundancy avoidance (RA) with respect to \lub{}s.

We first report upper bounds on the total number of \lub{}s associated
to a discrete set of points $N$.
Then we study the complexities of the RE and RA approaches.
Finally, we present some computational experiments that compare
these approaches.

\subsection{Tight upper bound on the number of local upper bounds}
\label{sec:counting}

None of the incremental algorithms proposed in the literature,
even in the \GP{} case, 
make it possible to directly derive a non-trivial upper bound 
on the size of any upper bound set $U(N)$ for $p\geq 4$.

For $p=2$, the number of \lub{}s is clearly $|N|+1$ (see Example~\ref{ex:bi_obj}).
For $p=3$ we recall that \citet{DaeKla14} showed that it is upper bounded by
$2|N|+1$ and is exactly $2|N|+1$ in the \GP{} case.

For an arbitrary $p\geq 2$, \citet{KapRubShaVer08} provide a tight upper bound
on the size of $U(N)$.
Following \citet{BoiShaTagYvi98} who study
the complexity of a union of axis-parallel hypercubes, 
they show that the number of maximal empty orthants 
with respect to a stable set $N$ 
is $O(|N|^{\lfloor \frac{p}{2} \rfloor})$.
They also provide an instance for which this number 
is $\Omega(|N|^{\lfloor \frac{p}{2} \rfloor})$.
Therefore, and recalling that maximal empty orthants 
are in one-to-one correspondence with local upper bounds, 
$O(|N|^{\lfloor \frac{p}{2} \rfloor})$ is a tight upper bound 
on the total number of local upper bounds associated to a stable set $N$.

\subsection{Worst-case complexities of the algorithms}

In this section, we analyze the worst case behavior of the two approaches.
We consider the dimension $p$ of the problem as a fixed parameter.
The reference algorithm for the RE approach 
will be Algorithm~\ref{alg:basic_incr2} 
while the reference algorithm for the RA approach 
will be Algorithm~\ref{alg:enh_incr} in the \GP{} case, and
Algorithm~\ref{alg:enh_incr_ngp} in the general case.

\paragraph*{Common steps of both approaches} 
In the \GP{} case, both approaches first compute the set $A$ 
of \lub{}s whose associated search zones contain $\bz$.
This amounts to $|U(N)|$ dominance tests if $U(N)$ is stored 
as a simple linked list.
If $A$ is small compared to $U(N)$, it is possible 
to reduce the complexity of these operations. 
Actually, since the elements of $A$ are those \lub{}s located 
in the hyperrectangle $\prod_{j=1}^p (\bz_j, M)$,
they can be obtained by an orthogonal range query on the set $U(N)$
\citep[see][Chapter~5]{deBChevanOve08}. 
In the case $p=2$, $U(N)$ can be efficiently stored 
in a simple balanced binary search tree.
For $p\geq 3$, as in the case of the algorithm of \citet{KapRubShaVer08},
$U(N)$ can be stored in a dynamic $p$-dimensional range tree
\citep[see e.g.][]{WilLue85},
which allows insertions and deletions in $O(\log^{p} |U(N)|)$ time 
and orthogonal range queries in
$O(\log^{p} |U(N)|+|A|)$ time. 
We note that \emph{augmented dynamic range trees} 
\citep[Theorem 8]{MehNah90} lower the ``log'' factors 
to ${\log^{p-1} |U(N)|}{\log\log |U(N)|}$.

\paragraph{Remaining steps}

We assume that $p\geq 3$ 
since it can be easily seen that both approaches operate 
identically in the case $p=2$.
Both approaches consider $p|A|$ candidate \lub{}s.

We first consider the \GP{} case. 
We focus on the operations 
on which Algorithm~\ref{alg:basic_incr2} (RE approach) and 
Algorithm~\ref{alg:enh_incr} (RA approach) differ.
They correspond to
Steps~\ref{alg:basic_incr2:start}-\ref{alg:basic_incr2:filtering} (Algorithm~\ref{alg:basic_incr2})
and Steps~\ref{alg:enh_incr:first_step}-\ref{alg:enh_incr:final_step} (Algorithm~\ref{alg:enh_incr}),
and respectively involve sets $P_j$, $j\in\onetod $,
and $P$.

\begin{proposition}
The worst-case complexity of Steps~\ref{alg:basic_incr2:start}-\ref{alg:basic_incr2:filtering} in Algorithm~\ref{alg:basic_incr2}
is bounded by 
$O(|A|^2)$.
\end{proposition}
\begin{proof}
The complexity of these steps is dominated by the filtering 
(Steps~\ref{alg:basic_incr2:start_filtering} and \ref{alg:basic_incr2:filtering})
of each $P_j$,
$j\in\onetod $,
where $|P_j|=|A|$,
therefore the total time is $O(|A|^2)$.
\end{proof}

This can be reduced 
to $O(|A|\log |A|)$ in the case $p \in \{2, 3\}$ \citep{KunLucPre75} and 
$O(|A|\log^{p-3} |A| \log \log |A|)$ in the case $p \geq 4$ \citep{GabBenTar84}
using some specialized algorithms.

\begin{proposition}
The worst-case complexity of Steps~\ref{alg:enh_incr:first_step}-\ref{alg:enh_incr:final_step} in Algorithm~\ref{alg:enh_incr}
is bounded by 
$O(|A|)$.
\end{proposition}

\begin{proof}
In Algorithm~\ref{alg:enh_incr}, 
no additional dominance test is performed with the \lub{}s of $P$,
but the values $z_j^{\max}(u)$ need to be computed just before they are needed,
each of which takes 
constant time.
Also the references to the $p$ points that define each \lub{} have to be updated
which takes 
constant time
for each new upper bound.
The total time of these operations is thus $O(|A|)$.
\end{proof}

In the general case, the number of \lub{}s 
against which candidate \lub{}s have to be checked for dominance 
in the RE approach
just grows by an additional $|B|$. 
In the RA approach adapted to the general case, namely Algorithm~\ref{alg:enh_incr_ngp}, 
it is possible that $|N|$ points have to be considered 
in a set $Z^k(u)$
at Step~\ref{alg:enh_incr_ngp:Zk}.
This leads to an upper bound on the complexity of 
$O(|N||A|)$ 
in Algorithm~\ref{alg:enh_incr_ngp}.

In practice, however, the size of the sets $Z^k(u)$ is rather small depending on
how many points in $N$ share the same component values.
Note that according to \citet{BoiShaTagYvi98}, an alternative approach would be
to slightly shift those points in $N$ that do not satisfy \GP{} such that
the resulting set satisfies \GP{}.
Then Algorithm~\ref{alg:enh_incr} can be applied, yielding a complexity of 
$O(|A|)$. 
Similarly, ties in the comparisons 
of any $j$th component values for points $z^k$ and $z^l$
could be resolved by a lexicographic comparison ``$<_{\rm lex}$'' where
$z^k_j <_{\rm lex} z^l_j$ if $z^k_j< z^l_j$ or if $z^k_j=z^l_j$ and $k<l$,
which would replace the natural comparison ``$<$'' between reals (and similarly for ``$>$'') 
in Algorithm~\ref{alg:enh_incr}.
However these approaches yield redundant search zones that, 
in the context of Algorithm~\ref{alg:basic_moco},
induce redundant solver calls.

\subsection{Experimental comparison of the algorithms}

In this section we investigate the behavior of the RE 
and RA approaches on random instances. 

\paragraph{Experimental setup}
We implemented Algorithm~\ref{alg:basic_incr2} for the RE approach
and Algorithm~\ref{alg:enh_incr_ngp} for the RA approach.
Both algorithms
were implemented in C.
The experiments were run on a workstation equipped with
an Intel Core i7-3840QM CPU at 2.80GHz with 8MB cache and
32GB RAM.
For both algorithms, we considered the version 
that does not require \GP{}, 
since the assumption cannot be made in most applications.

As test instances, we generated random stable
sets of points $N$. 
In order to obtain a new point in the random stable set being generated, 
we uniformly draw from the integer set $[1,K]^p$ 
and reject the points that are dominated by or dominate 
any of the previous points.
We draw without or with replacement in $[1,K]$, respectively,
to obtain points satisfying \GP{}
or not.
In the general case, 
the parameter $K$ controls to what extent
objective values are shared among feasible points.
In the \GP{} case, $K$ is just set 
to a very large integer.
Since in both cases the distribution of each point is conditioned 
by the requirement that it is neither dominated by nor dominates any previously generated point,
the generated points are eventually randomly reordered.

We considered instances for $p\in \{3,4,5,6\}$ having 
$100~000$, $50~000$, $25~000$, and $5~000$ points, respectively.
We generated instances under \GP{} and also with possible identical component values.
In the general case, we set $K$ so as to obtain $\frac{|N|}{K} \in \{5, 10\}$.
The plots we made in the \GP{} case were obtained by recording
intermediate results every 500 points for $p\in \{3,4,5\}$ 
and every 100 points for $p=6$.
We also considered a pathological instance type in the general
case having many duplicated component values among points, with $p = 6$, $|N| = 10~000$ and $K = 10$.

We have drawn 10 instances of each type 
and the output results were averaged over
the 10 runs carried out for each instance type.

\paragraph{Observations on the \GP{} instances}

We provided above a theoretical tight upper bound 
on the number of \lub{}s in \GP{} instances.
Now we consider the empirical number of \lub{}s 
observed in our test instances for $p\in \{4,5,6\}$ 
(since this number is known exactly for $p=3$).
The results, 
which can be obtained by any of the two approaches, 
are reported on Figure~\ref{fig:nlub}.
According to Figure~\ref{fig:nlub}, it seems that on such random instances, the number
of \lub{}s grows approximately linearly in the number of points.
\citet{KapRubShaVer08} show that the number of maximal empty axis-parallel boxes
in a set of $n$ points drawn uniformly and independently from $[0,1]^p$
is $O(n\log^{p-1}n)$, 
therefore the growth observed in our experiments may be superlinear.
However, the distribution of our points is not the same 
since we discard points that dominate
or are dominated by previously drawn points
and the bound of \citet{KapRubShaVer08} does not count only maximal empty orthants.

Observing from Figure~\ref{fig:nlub} the apparently linear relation between $|N|$ and $|U(N)|$,
we performed a simple linear regression.
We obtained the following slopes for the fitted lines: 6.524 for $p=4$, 31.86 for $p=5$,
and 165.9 for $p=6$. 
This gives an insight on the increase in the number of \lub{}s
induced by the consideration of a new point in the search region,
i.e. the average $|U(N\cup\{\bz\})|-|U(N)|$.

\begin{figure}[!htbp]
\begin{centering}
\includegraphics{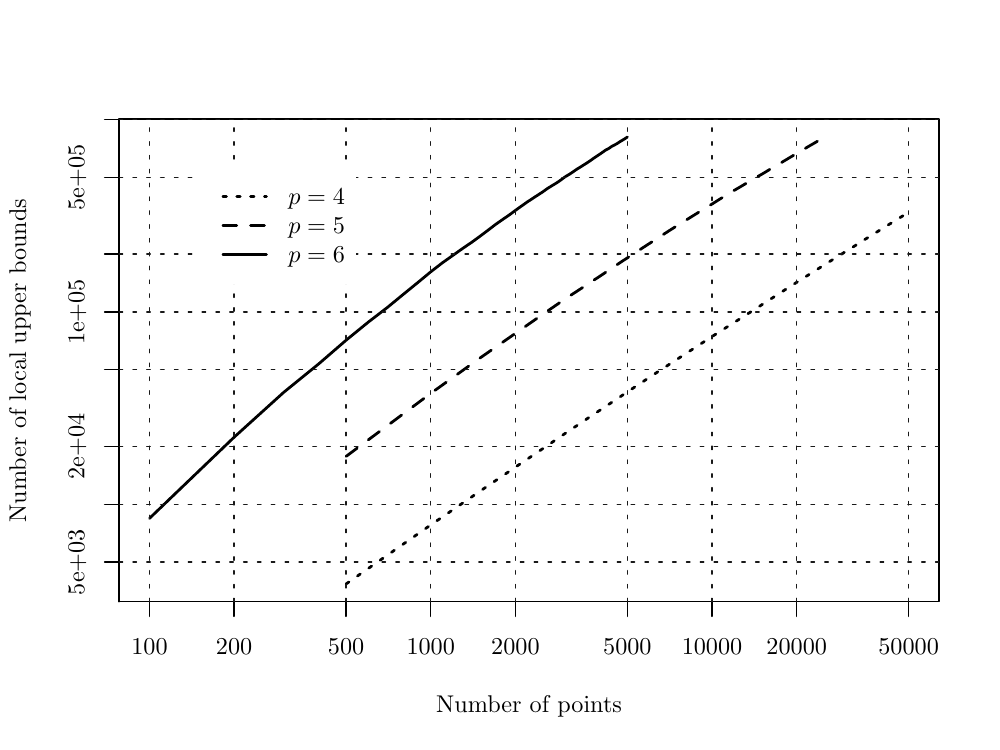}
\par\end{centering}
\caption{Number of \lub{}s on \GP{} instances 
(logarithmic scales for both axes)
\label{fig:nlub}}
\end{figure}

We also computed the average number of search zones that contain 
the current point (namely $|A|$) between two observations. 
Since these values do not vary much on the tested instance, 
we provide the averages over all instances of all sizes: 3.999 for $p=3$, 21.56 for $p=4$,
141.67 for $p=5$, and 735.9 for $p=6$.

From the average $|U(N\cup\{\bz\})|-|U(N)|$ and $|A|$, 
we compute the ratio $\frac{|U(N\cup\{\bz\})|-|U(N)|}{|A|}$.
We obtain 
$0.5001$ for $p=3$, $0.3025$ for $p=4$, $0.2249$ for $p=5$, and $0.2255$ for $p=6$.
This indicates that the number of additional search zones 
induced by each search zone that contains the current feasible point remains small.

\paragraph{Comparison of the algorithms}

We first provide raw computation times in Table~\ref{tab:times},
showing the performance of the RE and RA approaches on \GP{} and on general instances.

\begin{table}
\begin{center}
\subfloat[\GPl{}]{
\scalebox{0.9}{
\begin{tabular}{r@{\extracolsep{0pt}.}l|r@{\extracolsep{0pt}.}lr@{\extracolsep{0pt}.}lr@{\extracolsep{0pt}.}lr@{\extracolsep{0pt}.}l|r@{\extracolsep{0pt}.}lr@{\extracolsep{0pt}.}lr@{\extracolsep{0pt}.}lr@{\extracolsep{0pt}.}l}
\multicolumn{2}{c}{} & \multicolumn{8}{c}{RE approach} & \multicolumn{8}{c}{RA approach}\tabularnewline
\hline 
\multicolumn{2}{c|}{\backslashbox{$|N|$}{$p$}} & \multicolumn{2}{c}{3} & \multicolumn{2}{c}{4} & \multicolumn{2}{c}{5} & \multicolumn{2}{c|}{6} & \multicolumn{2}{c}{3} & \multicolumn{2}{c}{4} & \multicolumn{2}{c}{5} & \multicolumn{2}{c}{6}\tabularnewline
\hline 
\multicolumn{2}{r|}{5~000} & 0&232 & 1&09 & 26&3 & 614&0 & 0&289 & 2&18 & 35&2 & 179&0\tabularnewline
\multicolumn{2}{r|}{25~000} & 6&93 & 93&4 & 830&0 & \multicolumn{2}{c|}{-} & 16&6 & 154&0 & 951&0 & \multicolumn{2}{c}{-}\tabularnewline
\multicolumn{2}{r|}{50~000} & 44&6 & 509&0 & \multicolumn{2}{c}{-} & \multicolumn{2}{c|}{-} & 122&0 & 673&0 & \multicolumn{2}{c}{-} & \multicolumn{2}{c}{-}\tabularnewline
\multicolumn{2}{r|}{100~000} & 387&0 & \multicolumn{2}{c}{-} & \multicolumn{2}{c}{-} & \multicolumn{2}{c|}{-} & 664&0 & \multicolumn{2}{c}{-} & \multicolumn{2}{c}{-} & \multicolumn{2}{c}{-}\tabularnewline
\hline 
\end{tabular}}
}

\subfloat[\NGPl{}, $\frac{|N|}{K} = 5$]{
\scalebox{0.9}{
\begin{tabular}{r@{\extracolsep{0pt}.}l|r@{\extracolsep{0pt}.}lr@{\extracolsep{0pt}.}lr@{\extracolsep{0pt}.}lr@{\extracolsep{0pt}.}l|r@{\extracolsep{0pt}.}lr@{\extracolsep{0pt}.}lr@{\extracolsep{0pt}.}lr@{\extracolsep{0pt}.}l}
\multicolumn{2}{c}{} & \multicolumn{8}{c}{RE approach} & \multicolumn{8}{c}{RA approach}\tabularnewline
\hline 
\multicolumn{2}{c|}{\backslashbox{$|N|$}{$p$}} & \multicolumn{2}{c}{3} & \multicolumn{2}{c}{4} & \multicolumn{2}{c}{5} & \multicolumn{2}{c|}{6} & \multicolumn{2}{c}{3} & \multicolumn{2}{c}{4} & \multicolumn{2}{c}{5} & \multicolumn{2}{c}{6}\tabularnewline
\hline 
\multicolumn{2}{r|}{5~000} & 0&192 & 0&882 & 13&4 & 530&0 & 0&24 & 1&15 & 18&7 & 166&0\tabularnewline
\multicolumn{2}{r|}{25~000} & 6&18 & 68&3 & 767&0 & \multicolumn{2}{c|}{-} & 9&53 & 112&0 & 862&0 & \multicolumn{2}{c}{-}\tabularnewline
\multicolumn{2}{r|}{50~000} & 36&0 & 463&0 & \multicolumn{2}{c}{-} & \multicolumn{2}{c|}{-} & 77&0 & 582&0 & \multicolumn{2}{c}{-} & \multicolumn{2}{c}{-}\tabularnewline
\multicolumn{2}{r|}{100~000} & 339&0 & \multicolumn{2}{c}{-} & \multicolumn{2}{c}{-} & \multicolumn{2}{c|}{-} & 498&0 & \multicolumn{2}{c}{-} & \multicolumn{2}{c}{-} & \multicolumn{2}{c}{-}\tabularnewline
\hline 
\end{tabular}}}

\subfloat[\NGPl{}, $\frac{|N|}{K} = 10$]{
\scalebox{0.9}{
\begin{tabular}{r@{\extracolsep{0pt}.}l|r@{\extracolsep{0pt}.}lr@{\extracolsep{0pt}.}lr@{\extracolsep{0pt}.}lr@{\extracolsep{0pt}.}l|r@{\extracolsep{0pt}.}lr@{\extracolsep{0pt}.}lr@{\extracolsep{0pt}.}lr@{\extracolsep{0pt}.}l}
\multicolumn{2}{c}{} & \multicolumn{8}{c}{RE approach} & \multicolumn{8}{c}{RA approach}\tabularnewline
\hline 
\multicolumn{2}{c|}{\backslashbox{$|N|$}{$p$}} & \multicolumn{2}{c}{3} & \multicolumn{2}{c}{4} & \multicolumn{2}{c}{5} & \multicolumn{2}{c|}{6} & \multicolumn{2}{c}{3} & \multicolumn{2}{c}{4} & \multicolumn{2}{c}{5} & \multicolumn{2}{c}{6}\tabularnewline
\hline 
\multicolumn{2}{r|}{5~000} & 0&167 & 0&772 & 12&1 & 447&0 & 0&218 & 0&978 & 16&7 & 150&0\tabularnewline
\multicolumn{2}{r|}{25~000} & 5&99 & 65&6 & 754&0 & \multicolumn{2}{c|}{-} & 7&84 & 107&0 & 839&0 & \multicolumn{2}{c}{-}\tabularnewline
\multicolumn{2}{r|}{50~000} & 32&8 & 447&0 & \multicolumn{2}{c}{-} & \multicolumn{2}{c|}{-} & 69&2 & 564&0 & \multicolumn{2}{c}{-} & \multicolumn{2}{c}{-}\tabularnewline
\multicolumn{2}{r|}{100~000} & 325&0 & \multicolumn{2}{c}{-} & \multicolumn{2}{c}{-} & \multicolumn{2}{c|}{-} & 459&0 & \multicolumn{2}{c}{-} & \multicolumn{2}{c}{-} & \multicolumn{2}{c}{-}\tabularnewline
\hline 
\end{tabular}}}
\end{center}
 \caption{Average computation times (in seconds) for both approaches
 \label{tab:times}}
\end{table}

Since the computation times of the algorithms we consider 
to generate $U(N\cup\{\bz\})$
are both $\Omega (|U(N)|)$,
we also present normalized computation times.
Figure~\ref{fig:comp_time_gp} shows running times divided
by $|U(N)|$.

\begin{figure}[!htbp]
\begin{centering}
\subfloat[$p=3$]{\includegraphics[width=.5\textwidth]{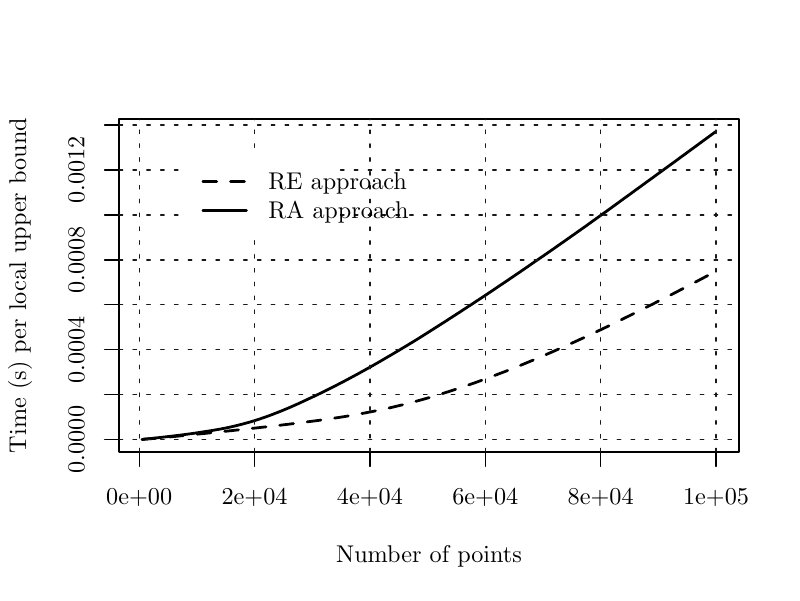}}
\subfloat[$p=4$]{\includegraphics[width=.5\textwidth]{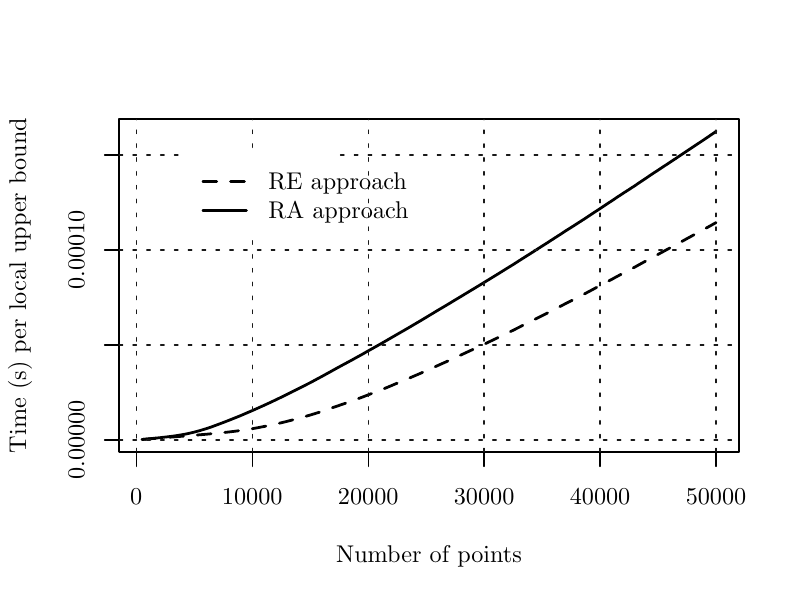}}\\
\subfloat[$p=5$]{\includegraphics[width=.5\textwidth]{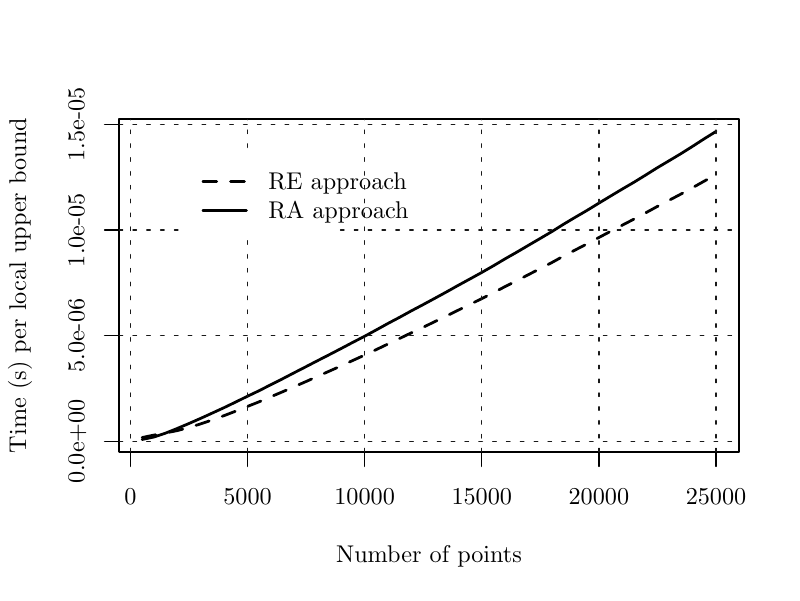}}
\subfloat[$p=6$]{\includegraphics[width=.5\textwidth]{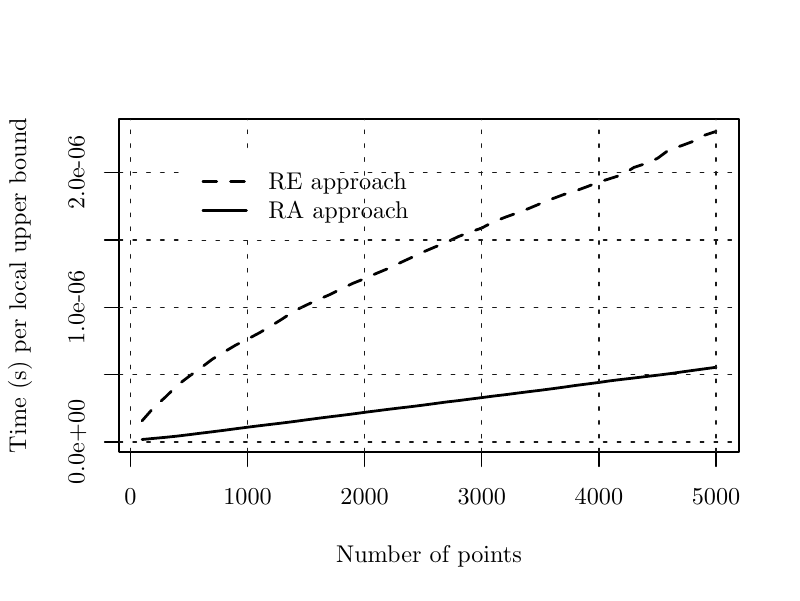}}
\end{centering}
\caption{Comparison of the normalized running times of the two algorithms in the \GP{} case
\label{fig:comp_time_gp}}
\end{figure}

According to these results, the RE approach remains the best one
in terms of computation time 
for $p \in \{3,4,5\}$, the values being rather close in the case $p = 5$.
The RA approach however outperforms the RE approach for $p=6$.
These observations hold for \GP{} and general case instances 
but the gaps between the relative efficiency of the approaches are larger on \GP{} instances.
Besides, additional computational experiments on \GP{} instances with $p\in \{7, 8\}$ 
showed that the RA approach performs even better above $p=6$.
Namely, we obtained the following average computation times (RE time in seconds, RA time in seconds):
($161.69$, $5.046$) for $p=7$, $n=500$, 
($1~240.11$, $27.458$) for $p=7$, $n=1~000$, 
($24.06$, $0.348$) for $p=8$, $n=125$,
($423.15$, $2.330$) for $p=8$, $n=250$.

We also observed in our experiments that, even starting from $p\geq3$,
a little fewer component comparisons 
are made in the RA approach than in the RE approach.
Finally, we ran the implementations under Cachegrind, a CPU caches profiling tool.
We observed for the RA approach a larger use of the slowest caches, L2 and L3,
than for the RE approach.
This, together with the smaller average $|A|$ observed 
on low dimensional instances, 
explains why the implementation of the RA approach 
performs worse than the one of the RE approach for $p \leq 5$.

To observe the effect of highly duplicated component values among distinct points,
we also tested the approaches on the pathological instances 
($p=6$, $|N|=10~000$, and $K=10$).
The average $|U(N)|$ and $|A|$ are much smaller 
than in the \GP{} case, being respectively $14~228.4$ and 33.44.
Due to the fact that many points share the same component values,
the sets $Z^k(u)$ in Algorithm,\ref{alg:enh_incr_ngp} 
can grow significantly,
reaching the maximum value of $1~109.7$, averaged on the test instances.
Therefore, the computation time of the RA approach is a little longer 
than the one of the RE approach (1.95 against 1.51 seconds).

\section{Conclusions}\label{sec:conclusions}

We addressed in this {\DocType} the problem of representing the search region 
in MOO.
The concept itself is used in numerous approaches 
to compute the nondominated set.
We provided several equivalent definitions of the search region.
Local upper bounds induce a decomposition of the search region 
into search zones.
We reviewed possible uses of this decomposition
to enumerate all nondominated points of an MOCO problem.
We presented two incremental approaches to compute the \lub{}s that represent 
a search region, respectively based on ``redundancy elimination (RE)''
and ``redundancy avoidance (RA)''. 
The first encompasses an
already known algorithm for which we proposed 
some enhancements to its filtering step. 
The second is derived from theoretical properties of \lub{}s we studied
and avoids the filtering step of the former.
Finally, we considered the complexity of the representation of the search region
by \lub{}s and gave some insights into the theoretical complexities
and the practical efficiencies of the two incremental approaches.
In particular, we showed that the RA approach developed in this {\DocType}
performs better than the RE approach starting from dimension 6
on instances where the objective ranges are not too small.

The future work directions are numerous.
Although the RA approach is practically less efficient than the RE approach 
in low dimensions, it maintains, contrary to the latter, 
a relation between feasible points and \lub{}s.
This makes it possible to define a neighborhood between \lub{}s, 
as in \citet{DaeKla14} in the case $p=3$,
that can be exploited in order to update the search region more efficiently
when a search zone containing the new feasible point is known.
Derivatives of the concept of search region defined in this {\DocType} 
could also be considered. 
Actually, we made no assumption on the feasible
points that define the search region, apart from the requirement
that they constitute a stable set of points.
If e.g. the feasible points are optimal with respect to one objective function,
some search zones may be discarded.
Note also that the search zones that are defined in this {\DocType}
are bounded below by the same point $\Bm$.
It may be interesting to bound below each search zone using some 
\emph{local lower bounds}
such that the union of the corresponding restricted search zones
still contains all unknown nondominated points.

\section*{Acknowledgments}

We acknowledge Carlos M. Fonseca from Universidade de Coimbra, Portugal, 
for pointing us to references describing the complexity 
of the upper bound set given in Section~\ref{sec:counting}.

\bibliographystyle{abbrvnat}
\bibliography{manuscript_arxiv}

\end{document}